\documentclass[a4paper]{article}
\usepackage[margin=25mm]{geometry}
\usepackage{graphicx}
\usepackage{subfigure}
\usepackage{float}
\usepackage{graphicx}
\usepackage{amssymb}
\usepackage{mathrsfs}
\usepackage{fullpage}
\usepackage{amsmath}
\usepackage{amsthm}
\usepackage{graphicx,caption}
\usepackage{eucal}
\usepackage{lipsum}
\usepackage{amsfonts}
\DeclareMathAlphabet{\mathpzc}{OT1}{pzc}{m}{it}
\usepackage{epsfig}
\usepackage{adjustbox}
\allowdisplaybreaks
\newtheorem{theorem}{Theorem}[section]
\usepackage{latexsym}
\usepackage[nodisplayskipstretch]{setspace}
\makeatletter
\def\ps@pprintTitle{%
\let\@oddhead\@empty
\let\@evenhead\@empty
\def\oddfoot\@empty
\let\evenfoot\@oddfoot}
\makeatother
\usepackage{verbatim}
\usepackage{lineno}
\providecommand{\keywords}[1]
{
  \small	
  \textbf{\textit{Keywords---}} #1
}
\title {Optimization of Traffic Control in $\textrm{\it MMAP[k]/PH[k]/S}$ Catastrophic Queueing Model with $\textit{P\!H}$ Retrial Times and Controllable Preemptive Repeat  Policy 
} 

\author{Raina Raj$^{1}$, Vidyottama Jain$^{1}$   \thanks{Corresponding author}   \\
        \small Central University of Rajasthan, Ajmer, India$^{1}$\\
    }
  
\date{} % Comment this line to show today's date

\begin{document}
\maketitle

\begin{abstract}
The presented study elaborates a multi-server catastrophic retrial queueing model considering  preemptive repeat priority policy with phase-type ($\textrm{\it P\!H}$) distributed retrial times. For the sake of comprehension, the scenario of model operation prior and later to the occurrence of the disaster is referred to as the normal scenario and as the catastrophic scenario, respectively. In the normal scenario, the incoming heterogeneous calls are categorized as handoff calls and new calls.  Handoff calls are provided  controllable preemptive priority over new calls. In the catastrophic scenario, when a disaster causes the shut down of the entire system and failure of all functioning channels,  a set of backup channels is quickly deployed to restore services.   Due to the emergency situation in the concerned area,  the incoming heterogeneous calls are divided into three categories: handoff, new call, and emergency calls. Emergency calls are provided controllable preemptive priority over new/handoff calls due to the pressing need to save lives in such situations.  The Markov chain's ergodicity criteria are established by demonstrating that it belongs to the class of asymptotically quasi-Toeplitz Markov chains $\textrm{\it (A\!Q\!T\!M\!C)}$.   Further, a multi-objective optimization problem to obtain optimal number of  backup channels  has been formulated and dealt by employing non-dominated sorting genetic algorithm-II (NSGA-II) approach.
 
\end{abstract} \hspace{10pt}

%TC:ignore
%\subclass{60K25, 68M20, 90B22, 90B18}

\section{Introduction} \label{section1}

With the advancement of technology, there has been observed a surge in queueing models for studying the dynamics of communication systems. The study of disastrous events that result in an abrupt change in the macroscopic state of an entire network has been an important aspect of this trend.
These disastrous events are either human-induced (e.g., fire, virus attack, power outages, etc.) or natural (e.g., flood, tsunami, cyclone, etc.), and are typically mentioned as catastrophic events. With the occurrence of a catastrophic event,  all active and waiting customers/traffic are compelled to exit the system abruptly, rendering it inoperable. For a comprehensive survey of catastrophic events occurring in  communication networks, refer \cite{dabrowski2015catastrophic} and references cited therein. For these types of catastrophic queueing models, prioritization of traffic is also important  in many cellular network applications,  e.g.,  voice traffic is provided priority over data traffic in multiprocessor switching. The system's various aspects, such as arrival discipline, service discipline, categories of services, and so on, ensure traffic prioritization. Many priority policies, such as guard channel policy, threshold policy, preemptive priority policy, non-preemptive priority policy, and others, have been proposed in the literature for this purpose (refer, \cite{brandwajn2017multi,chang1965preemptive,krishnamoorthy2008map,machihara1995bridge}). 
In this investigation,  preemptive repeat priority policy has been implemented for the higher priority traffic. When all of the channels are occupied, but at least one of these channels is occupied by a lower priority traffic, the arriving higher priority traffic will preempt the service of the ongoing lower priority traffic. This preempted lower priority traffic will  join a virtual space, named orbit, and will then retry for the service from the scratch. The process of recurrent attempts to obtain the service is referred as retrial phenomenon and  it has been extensively explored by the researchers in the area of communications and cellular networks (refer, \cite{kim2016survey}). Retrial phenomenon is ubiquitous in many real life applications  such as call centers, communication systems, optical networks, inventory systems, and so on. In modern cellular networks, retrial is just a matter of pushing one button. In some communication applications, blocked calls are automatically redialed. Therefore, when developing these systems, retrial phenomenon must be taken into account.

The layout of this work is  arranged in seven sections. Section \ref{section1a} represents the state of art related to the proposed work. In Section \ref{section2},  a  $\textrm{\it M\!M\!A\!P[c]/P\!H[c]/S}$ model with  $\textrm{\it P\!H}$ distributed retrial times  and catastrophe phenomenon is described.  In Section \ref{section3},  the infinitesimal generator matrix for the proposed $\textit{L\!D\!Q\!B\!D}$ process has been derived. The ergodicity condition of the underlying process is obtained  by proving that the Markov chain belongs to the class of $\textrm{\it A\!Q\!T\!M\!C}$. A modified algorithm is employed to compute steady-state probabilities.  In Section \ref{section4}, formulas of key performance measures to analyse the system efficiency  are derived. Numerical  illustrations to point out the impact of various intensities over the system performance are presented in Section \ref{section5}. An optimization problem has been formulated to evaluate the  behaviour of the system in Section \ref{section6}. Finally, the underlying model is concluded with the insight for the future works in Section \ref{section7}.

\section{State of the Art} \label{section1a}

 Challenges emerging from the consideration of the catastrophic  queueing models with or without the retrial phenomena have been highlighted in the following works.
  An overview of research works on the catastrophic queueing models without retrial phenomenon can be found in the articles (refer, \cite{baumann2012steady,sudhesh2017transient,yajima2019central,yechiali2007queues}). Baumann and Sandmann \cite{baumann2012steady}  proposed an efficient algorithm for the computation of stationary distribution for a $M/M/c$ queueing model with catastrophic events. Sudhesh et al.  \cite{sudhesh2017transient} considered a   two heterogeneous servers queueing model with catastrophic event, server failure and repair.  % In their work, the waiting customers were presumed to be impatient and may exit the system at any time, when the system was down. 
  Yajima and Phung-Duc \cite{yajima2019central} presented a $M^x/M/\infty$ queueing model with binomial catastrophe and established a central limit theorem for the stationary queue length of model in a heavy traffic regime. Yechiali \cite{yechiali2007queues}  studied a $M/M/c$ queueing system with disastrous events considering $c=1, 1<c<\infty, c=\infty.$   Zhou and Beard \cite{zhou2009controlled} proposed a multi-server catastrophic queueing model considering controlled preemption policy for the emergency calls.  However, the applicability of  the above mentioned models has been diminished in the present scenario, since the incoming call arrival followed Poisson process and service times were considered to be exponentially distributed. In contrast to the memory-less property of stationary Poisson flow, the input stream of arrival contains burstiness and correlation properties.    Thus, more generalized arrival and services processes are employed, such as Markovian arrival process ($\textit{M\!A\!P}$), marked Markovian arrival process ($\textit{M\!M\!A\!P}$) and phase-type ($\textit{P\!H}$) distribution.
   Some of the  pertinent studies with the consideration of more general processes  are as follows. 
    Chakravarthy \cite{chakravarthy2017catastrophic} presented a $\textit{M\!A\!P/P\!H/1}$ queueing model considering catastrophic event and delayed action. Recently, Kumar and Gupta \cite{kumar2020analysis} studied a discrete–time catastrophic model with population arrival following batch Bernouli process, and binomial catastrophe arrival occurred according to the discrete-time renewal process according to which each individual  either survived with probability $p$ or died with probability ($1-p$).

 Reviews of some of the relevant literature for the catastrophic retrial queueing models are as follows.
 Wang et al.  \cite{wang2008transient}  proposed a $M/G/1$ retrial queueing model with catastrophe phenomenon. In this study, the inter-retrial times were exponentially distributed and catastrophe arrival occurred according to Poisson process. On the similar track,  Chakrvarthy et al. \cite{chakravarthy2010retrial}  presented a $\textrm{\it M\!A\!P/P\!H/1}$ retrial queueing model with catastrophe phenomenon and repair process. The arrival of catastrophe followed Poisson process and failed channels were repaired following exponential distribution. Their model was a confined one due to the assumption of exponential distribution for retrial and repair processes. Recently,  Ammar and Rajadurai \cite{ammar2019performance} proposed a $M/G/1$ preemptive priority retrial queue with  working breakdown services and disasters. 
   Though the above mentioned studies had considered generalized arrival and/or service processes, yet retrial process was described through exponential distribution only. 
   
   In wireless cellular networks, the inter-retrial times are notably brief in comparison to the service times. Since, the retrial attempt is just a matter of pushing one button, these retrial customers will make numerous attempts during any given service interval. Therefore, the consideration of exponential retrial times in place of non-exponential ones could lead to under/over estimating the system parameters as shown by various studies in the literature (refer, \cite{chakravarthy2020retrial,dharmaraja2008phase,jain2021,shin2011approximation}). Raj and Jain \cite{raj2021} presented a $\textrm{\it M\!M\!A\!P[2]/P\!H[2]/S}$ queueing model with $\textrm{\it P\!H}$ distributed retrial times and preemptive repeat priority policy. They proposed a traffic control optimization problem and employed heuristic approaches to obtain its optimal solution. 

The fundamental impetus for studying a  multi-server catastrophic model with $\textrm{\it P\!H}$ distributed retrial times stemmed from the requirement to include non-exponential retrial times. The existing literature dealt with  non-exponential retrial times with some additional constraint such as constant retrial rate, two state $\textrm{\it P\!H}$ distribution, priority of customers in the orbit, exponential distributed service times, Poisson input flow, etc. This work introduces a  model with   $\textrm{\it P\!H}$ distributed retrial times with no such constraints.
Additionally, the proposed model can be mapped to the earlier reported models having fixed phases in retrial times (refer, \cite{shin2011approximation}). 
To the best of authors' knowledge, the proposed model is the first one that deals with such complex system considering the preemptive repeat priority policy.  Therefore, this work fills the gap in the literature by taking into account $\textrm{\it P\!H}$ distributed inter-retrial times in  the catastrophic model. Furthermore, this approach introduces the concept of backup/standby channels to deal with a disaster. This paper presents a model in cellular networks which represents the relevance of  arrival, service times and retrial times with the realistic scenario. Besides  the obvious applications in cellular networks, the suggested model has compelling applications in call centres, local area networks, computer communication systems, cellular networks and other environments where virus attacks, natural calamities, fire and other disasters render the system inoperable.

Motivated by these factors, in this work,  a  multi-server catastrophic  queueing model with  $\textrm{\it P\!H}$ distribution for retrial process and preemptive repeat priority policy is introduced.   To the best of authors' knowledge, the proposed model is the first one that deals with such complex system. 
For the sake of clarity, the scenario of model functioning prior to the occurrence of  a disaster (man-made/natural) is referred to as the normal scenario, and after the disaster, as the catastrophic scenario. In the normal scenario, the system will  provide services to all the incoming calls; however, in the catastrophic scenario, the system will collapse, flushing out all active and waiting calls.  
In normal scenario, the incoming calls are classified as handoff call and new call.  The new call, which finds all the channels busy upon its arrival will join the orbit of an infinite capacity and will be referred as a retrial call (see, \cite{jain2020numerical}). The retrial call can either retry for service or quit the system without receiving the service. When all of the channels are occupied and at least one of them is occupied with a new call, an arriving handoff call is given preemptive priority over the ongoing new call. The arriving handoff calls can preempt the service of  ongoing new calls up to a threshold level $K_2~(0\leq K_2 \leq S)$; otherwise, the arriving handoff call is considered lost from the system. The preempted new call will join the orbit and retry for its service from the scratch.

  In the catastrophic scenario, when a calamity strikes, the  operating system is entirely shut down. Due to the sudden outbreak, all active and retrial calls are forced to terminate their processes and exit  the system.  Meanwhile, when all of the channels fail in the system, a set of backup/standby channels is instantly installed in the affected area, and services are immediately restored. Calls to and from emergency services, such as hospitals, police, and fire departments, should be given precedence over other public calls in such tragic circumstances. Thus, the incoming heterogeneous calls are now classified as handoff calls, new calls and emergency calls.   An arriving emergency call will be dropped, if all the channels are occupied by emergency calls only, otherwise it will be  given  preemptive priority over  new calls/ handoff calls.  But, the preemption will occur until the active number of emergency calls are less than $K_1~ (0\leq K_1 \leq S)$.
  The arriving emergency call will start receiving service in place of the preempted new/handoff call. These preempted calls will be lost from the system as it is not reasonable to keep the concept of orbit in catastrophic scenario.  The underlying process of this presented system is modeled by level dependent quasi-birth-death ($\textit{L\!D\!Q\!B\!D}$) process.
The detailed study over   $\textit{L\!D\!Q\!B\!D}$ process can be found in \cite{latouche1999introduction}  and  \cite{he2014fundamentals}.
Ergodicity conditions of the underlying Markov chain are obtained by proving that the Markov chain satisfies the properties of asymptotically quasi-Toeplitz Markov chains ($\textrm{\it A\!Q\!T\!M\!C}$) (see, \cite{klimenok2006multi}). 
 An algorithmic approach, which was earlier proposed by \cite{dudin2019retrial}, is modified for the efficient computation of the steady-state distribution. In addition, for the numerical illustration, the expressions of key performance measures have been derived.  Due to the  consideration of the preemptive  priority policy, the blocking probability for emergency calls decreases and simultaneously the frequent termination of services for handoff and new calls increases. With a careful observation, it is realized that one of the important aspects is to determine the optimal number of the required backup channels in  that particular scenario. Therefore, a multi-objective optimization problem to obtain optimal value of total number of  backup channels  has been formulated such that the loss probabilities should not exceed some pre-defined values. Further, it is dealt by employing non-dominated sorting genetic algorithm-II (NSGA-II) approach (see, \cite{deb2002fast}).

 In this work, the construction of a multi-dimensional Markov chain (MDMC) describing the dynamics of the system, including the proper choice of the components defining the service and retrial processes, is required for the analysis of a multi-server queuing system with  $\textrm{\it P\!H}$ distributed  service and retrial times. There are two different ways, referred to as the TPFS (track-phase-for-server) and CSFP (count-server-for-phase) (\cite{he2018space}),  to keep track of the phases for service and retrial processes. In this work, CSFP approach is employed for the sake of performing computation in an efficient way since it provides comparatively  much smaller state space than TSFP.

\section{Model Description} \label{section2}
\subsection{Details and Assumptions}

The presented study introduces  a $\textrm{\it M\!M\!A\!P[c]/P\!H[c]/S}$ catastrophic model with $\textrm{\it P\!H}$ distributed retrial times. %Here, $S$ is defined as the total number of channels in the system. 
The proposed approach can be observed to work in two  scenarios: normal and catastrophic. Due to the classification and prioritization of incoming calls, it has been shown that the proposed model works in two scenarios. The model's operation prior to the onset of a disaster can be described as normal, and the latter as catastrophic.

\begin{itemize}
    \item[-]\textbf{ Normal Scenario:}   
    In this scenario,  the incoming heterogeneous calls are categorized in two classes ($c=2$) as handoff calls and new calls. The arrival and service processes of both handoff and new calls follow $\textrm{\it M\!M\!A\!P}$ and $\textrm{\it P\!H}$ distribution with distinct parameters, respectively. If all of the channels are occupied when a new call arrives, the new call will join the orbit (virtual space) of an infinite capacity and will be referred to as a retrial call (refer, \cite{jain2020numerical}).
     The retrial call following $\textit{P\!H}$ distribution can either retry for service or exit the system without obtaining the service. It has been considered that when the number of retrial calls $\mathpzc{l}$ is between 0 and $M$, the retrial rate is $\theta_{\mathpzc{l}}$ and once the number exceeds $M,$ the retrial rate is considered as $\theta.$ 
     %%%%%%%%%%
      When a handoff call arrives to the system and all available channels are occupied, one of the following two cases may occur. If all the channels are occupied by handoff calls, the incoming handoff call will be lost from the system, else it will be given preemptive priority over the new call in service.  Consequently, the handoff call will commence its service in place of the preempted new call and this preempted new call  will join the orbit.  The handoff call can preempt the service of an ongoing new call up to a threshold level $K_2~(0\leq K_2 \leq S).$ An arriving handoff call can preempt the service of an ongoing new call until the number of active handoff calls are less than $K_2$ in the system. If the number of active handoff calls are greater than or equal to $K_2$, there will be no preemption.

\item[-] \textbf{ Catastrophic Scenario:}   The presented model is subjected to catastrophic events such as power outages, virus assaults, natural disasters, fires, and so on.  With the occurrence of such disaster events, all the calls in the system (the one in service and the one waiting for service) leave the system prematurely  and  the system becomes inactivated.  Since a catastrophic event is bursty by nature, $\textit{M\!A\!P}$ is an appropriate representation of the disaster's arrival phase. The failed channels  are immediately  repaired following $\textit{P\!H}$ distribution.
In this study, it is considered when the whole system is collapsed and all channels are failed, $K$ backup/standby channels will start providing services at a slow rate. These backup channels will stop working when at least one of the failed channels is repaired. Arriving calls are deemed lost from this point onward until all of the channels are fixed.
Since, the occurrence of a disaster causes an emergency situation in the affected area,  the incoming heterogeneous calls are now categorized in three classes ($c=3$) as handoff calls, new calls and emergency calls. The arrival and service processes of all types of calls follow $\textit{M\!M\!A\!P}$ and  $\textit{P\!H}$ distributions with distinct parameters, respectively.  To provide emergency services in catastrophic scenario, a general sense of priority should be attached to the  emergency calls. Here, it is assumed that the emergency calls are provided  controllable preemptive priority over the handoff and new calls. When an emergency call arrives to the system, out of the following three cases, one might occur.
 \begin{itemize}
     \item[1.] When an arriving emergency call finds  all of the channels  occupied with the emergency calls only,  the arriving emergency call will be lost from the system.
     \item[2.] If at the arrival epoch of an emergency call, all of the channels are occupied with handoff  calls only, the service of a handoff call will be preempted and the emergency call will start receiving service in place of that preempted handoff call.
     \item[3.] If at the arrival epoch of an emergency call,  all of the channels are occupied and at least one of the channels is occupied with a new call, the emergency call will preempt the service of that new call and commence service in its place.
 \end{itemize}
 But, the preemption will occur until the active number of emergency calls are less than $K_1~ (0\leq K_1 \leq K).$ If the number of active emergency calls are greater than or equal to $K_1$, there will be no preemption of either types of calls.
     It has already been mentioned  that there is no orbit for the blocked or preempted calls when backup channels are providing service. The reason is that, the calls in the orbit need to wait some time before receiving service and it is not appropriate to make calls wait for a significant time before being admitted when there are urgent needs to save lives or properties. Further,  all the required notations are described in Table \ref{tab:table1}.

\end{itemize}

\begin{table}[]
	\centering
	\scalebox{0.7}
	{
	\begin{tabular}{|l|l|}
	\hline
	$S$ & total number of channels.\\
	\hline
	$K$ & total number of backup channels.\\
	\hline
		$K_1$ & threshold value for controlled preemption in catastrophic scenarios\\
		\hline
		$K_2$ & threshold value for controlled preemption in normal scenarios\\
	\hline
	$C_0, C_{\mathcal{N}}, C_{\mathcal{H}}, C_{\mathcal{E}}$ & the square matrices of order $L_1$ that characterize the \textit{M\!M\!A\!P}.\\
	\hline
	$\lambda_{\mathcal{H}}, \lambda_{\mathcal{N}}, \lambda_{\mathcal{E}}$ &  the average arrival rates of handoff, new and emergency calls, respectively.\\
	\hline
	$(\beta_{\mathcal{H}}, A_{\mathcal{H}})$ & representation of  \textit{$P\!H$} distribution of handoff call with order $M_{\mathcal{H}}$.\\
	\hline
		$(\beta_{\mathcal{N}}, A_{\mathcal{N}})$ & representation of  \textit{$P\!H$} distribution of new call with order $M_{\mathcal{N}}$.\\
	\hline
		$(\beta_{\mathcal{E}}, A_{\mathcal{E}})$ & representation of  \textit{$P\!H$} distribution of emergency call with order $M_{\mathcal{E}}$.\\
	\hline
	$\mu_{\mathcal{H}},\mu_{\mathcal{N}},\mu_{\mathcal{E}}$ & the average service rates of handoff, new and emergency calls, respectively.\\
	\hline
		$(\gamma, \Gamma)$ & representation of  \textit{$P\!H$} distribution of retrial call with dimension $N$.\\
			\hline
			$\Gamma^0(1)$, $\Gamma^0(2)$ &  the absorption due to departure from the cell and   the absorption due to retrial attempt.\\
				\hline
				$ \theta$ & the average retrial rate of retrial call.\\
				\hline
			$D_0, D_{1}$ & the square matrices of size $L_2$ that characterize the \textit{M\!A\!P}.\\
			\hline
			$(\alpha, B)$ & representation of  \textit{$P\!H$} distribution of repair process with dimension $R$.\\
			\hline
			$T_n^m$ & $(m+n-1)!/(m-1)!n!$.\\
			\hline
			$\Tilde{A} $ & $\begin{pmatrix}
			0 & 0\\
			A^0 & A
			\end{pmatrix}$.\\
			\hline
			$P_{\kappa_1}(\beta_{\mathcal{H}})$ &  the matrix that defines the transition probabilities of the process at the epoch of starting  that $\kappa_1$ \\ &   channels are busy (refer, \cite{klimenok2006multi}).\\
			\hline
			$L_{S-\kappa_1-\kappa_2}(S-\kappa_2,\Tilde{A_H})$ &  the matrix that defines the transition intensities of the process  at the service completion epoch given that \\ & $\kappa_1$ handoff calls are busy  (refer, \cite{klimenok2006multi}).\\
			\hline
			$A_{\kappa_1}(S-\kappa_2,A_H)$ & the matrix that defines the transition intensities of the process which do not lead to the service completion \\ & given that  $\kappa_1$ handoff calls are busy (refer, \cite{klimenok2006multi}).\\
			\hline
			$\bigotimes$ and $\bigoplus$ &  the Kronecker product and Kronecker sum of matrices, respectively (refer, \cite{dayar2012analyzing}). \\
			\hline
			diag & main diagonal of a matrix.\\
			\hline 
			$\text{diag}^{+}$ & upper diagonal of a matrix.\\
			\hline
			$\text{diag}^{-}$ & lower diagonal of a matrix.\\
			\hline
			row & row vector.\\
			\hline
			col & column vector.\\
			\hline
	\end{tabular}}
	\caption{Notations}
	\label{tab:table1}
\end{table}

\section{Mathematical Analysis} \label{section3}

% i = \mathpzc{l}; j = \kappa; k= \mathfrak{j}; l = \mathfrak{i} 

The underlying process \{$\Xi(t), t \geq 0 \}$ for a cell is defined by the following state space:
\begin{align*}
	\nonumber \Omega &= \{(\mathpzc{l}, \kappa_1, \kappa_2,  \mathfrak{j}, i, v_1, v_2, s_{\mathcal{H}}^1, s_{\mathcal{H}}^2,\ldots,s_{\mathcal{H}}^{M_\mathcal{H}}, s_{\mathcal{N}}^1, s_{\mathcal{N}}^2,\ldots,s_{\mathcal{N}}^{M_\mathcal{N}}, s_{\mathcal{E}}^1, s_{\mathcal{E}}^2,\ldots,s_{\mathcal{E}}^{M_\mathcal{E}}, r^1, r^2,\ldots, r^{N}, v^1,v^2, \\ &~~~~~~ \ldots, v^{R});  \mathpzc{l} \geq 0,~0 \leq \kappa_1 \leq S,~0 \leq \kappa_2 \leq S,~ 0 \leq \mathfrak{j} \leq S,~0 \leq i \leq K,~1 \leq v_1 \leq L_1,~1 \leq v_2 \leq L_2\},\end{align*}
where, 
\begin{itemize}
	\item $\mathpzc{l}$ is the number of retrial calls,
	\item $\kappa_1$ is the number of handoff calls in the system receiving service,
	\item $\kappa_2$ is the number of new calls in the system receiving service,
	\item $\mathfrak{j}$ is the number of channels under repair,
	\item $i$ is the number of emergency calls in the system receiving service,
	\item $v_1$ is the current phase of $\textrm{\it M\!M\!A\!P}$ for arrival of calls,
		\item $v_2$ is the current phase  of $\textrm{\it M\!A\!P}$ for arrival of catastrophe,
	\item $s_{\mathcal{H}}^{m_1}$ is the number of channels for handoff calls which are in phase $m_1$; $s_{\mathcal{H}}^{m_1} = \overline{0,S},$ $m_1 = \overline{1,M_\mathcal{H}}$,
	\item $s_{\mathcal{N}}^{m_2}$ is the number of channels for new calls which are in phase $m_2$; $s_{\mathcal{N}}^{m_2} = \overline{0,S},$ $m_2 = \overline{1,M_\mathcal{N}}$,
	\item $s_{\mathcal{E}}^{m_3}$ is the number of channels for emergency calls which are in phase $m_3$; $s_{\mathcal{E}}^{m_3} = \overline{0,K},$ $m_3 = \overline{1,M_\mathcal{E}}$,
	\item $v^{h_2}$ is the number of channels under repair which are in phase $h_2$; $v^{h_2} = \overline{0,S},$ $h_2 = \overline{1,R}$,
	\item $r^{h_1}$ is the number of retrial calls  which are in phase $h_1$; $r^{h_1} \geq 0,$ $h_1 = \overline{1,N}$.
	\end{itemize}

 	The stochastic process \{$\Xi(t), t \geq 0 \}$ can be modelled as a level-dependent quasi-birth death  ($\textrm{\it L\!D\!Q\!B\!D}$) process with the  infinitesimal generator matrix provided as follows:
\begin{center}
	$\mathscr{Q} =
	\begin{pmatrix}
	\mathscr{Q}^{0} & \mathscr{Q}_{0,1} & 0 & 0 & 0 & 0 &  \\
	\mathscr{Q}^{'}& \mathscr{Q}_{1,1}& \mathscr{Q}_{1,2} & 0 & 0 & 0 &  \\
	\mathscr{Q}_{2,0}  & \mathscr{Q}_{2,1} & \mathscr{Q}_{2,2} & \mathscr{Q}_{2,3} & 0 &  0 &  \\
	\mathscr{Q}_{3,0}  & 0 &\mathscr{Q}_{3,2} & \mathscr{Q}_{3,3} & \mathscr{Q}_{3,4} &   0 &  \\
%		\mathscr{Q}_{4,0}  & 0 &0 &\mathscr{Q}_{4,3} & \mathscr{Q}_{4,4} & \mathscr{Q}_{4,5}     \\
	\vdots & \vdots  & \vdots & &\ddots & \ddots & \ddots \\
	\vdots& \vdots  & \vdots & & & \ddots & \ddots & \ddots \\
		\mathscr{Q}_{M,0}  & 0 &0 &&& \mathscr{Q}_{M,M-1} &\mathscr{Q}_{M,M} & \mathscr{Q}_{2} &  &   \\
%\mathscr{Q}^{+}  & 0 &0 &&& &\mathscr{Q}_{M,M-1} & \mathscr{Q}_{M,M} &        \\
\mathscr{Q}^{+}  & 0 &0 &&&& \mathscr{Q}_{0} & \mathscr{Q}_{1} & \mathscr{Q}_{2}    &   \\
\vdots  & \vdots & \vdots &&&& &\ddots & \ddots & \ddots       \\

	\end{pmatrix}.$
\end{center}
%=========================================

{\small{
\begin{align*}
& \textbf {Upper  Diagonal :}\\
& \mathscr{Q}_{\mathpzc{l},\mathpzc{l}+1} = \text{diag}\{ X_\mathpzc{l}(0),X_\mathpzc{l}(1), \ldots, X_{\mathpzc{l}}(S)\} + \text{diag}^+\{ \hat{X}_\mathpzc{l}(0), \hat{X}_\mathpzc{l}(1),\ldots,\hat{X}_{\mathpzc{l}}(S-1)\};\forall \mathpzc{l} \geq 0,\\
	&  X_\mathpzc{l}(\kappa_1) = \textrm{diag}\{X_\mathpzc{l}(\kappa_1,\kappa_2)\};\forall \kappa_1 = \overline{0,S}, \kappa_2 = \overline{0,S-\kappa_1},\\
  & X_0(\kappa_1,\kappa_2) = 
  \begin{cases}
%    (X_0(0,0,0), X_0(0,0,1), \ldots, X_0(0,0,S))^T;  \kappa_1 = \kappa_2=0,\\
\textrm{col}(X_0(0,0,j));\forall \mathfrak{j} = \overline{0,S} ,\\
      \textrm{col}(X_0(\kappa_1,\kappa_2,0),X_0(\kappa_1,\kappa_2,S));  \forall S=K \text{or}~ K<S ~~\&~~ \kappa_1+\kappa_2 \leq K,\\
      X_0(\kappa_1,\kappa_2,0);  \text{$\forall K<S ~\&~ \kappa_1+\kappa_2 > K$,}
  \end{cases} \\
  &X_0(\kappa_1,\kappa_2,\mathfrak{j}) = \begin{cases}
X_0(\kappa_1,\kappa_2,\mathfrak{j},0); 
 \text{$\forall \mathfrak{j} = \overline{0,S-1}$},\\
\textrm{col}(X_0(\kappa_1,\kappa_2,S,i));\forall i=\overline{0,K-\kappa_1-\kappa_2},  \end{cases} \\
  & \hat{X}_\mathpzc{l}(\kappa_1)  =
  \begin{pmatrix}
      \hat{X}_\mathpzc{l}(\kappa_1,0)&  & \\
      \hat{X}_\mathpzc{l}(\kappa_1,1)&  & \\
   %   & \hat{X}_\mathpzc{l}(\kappa_1,2)&  \\
      & \ddots & \\
       & & \\
      & & \hat{X}_\mathpzc{l}(\kappa_1,S-\kappa_1)
  \end{pmatrix}; \forall \kappa_1 = \overline{0,S-1},\\
 &  \hat{X}_0(\kappa_1,\kappa_2) = 
  \begin{cases}
  \textrm{col}(\hat{X}_0(0,0,\mathfrak{j})); \forall \mathfrak{j}=\overline{0,S},\\
      \textrm{col}(\hat{X}_0(\kappa_1,\kappa_2,0),\hat{X}_0(\kappa_1,\kappa_2,S));  \forall S=K  ~\text{or}~~ K<S  ~\&~ \kappa_1+\kappa_2 \leq K,\\
      \hat{X}_0(\kappa_1,\kappa_2,0);  \text{$\forall K<S ~\&~ \kappa_1+\kappa_2 > K$,}
  \end{cases}\\
  & \hat{X}_0(\kappa_1,\kappa_2,\mathfrak{j}) = \begin{cases}
  \hat{X}_0(\kappa_1,\kappa_2,\mathfrak{j},0); 
 \text{$\forall \mathfrak{j} = \overline{0,S-1}$},\\
\textrm{col}(\hat{X}_0(\kappa_1,\kappa_2,S,i)); \forall  i = \overline{0,K-\kappa_1-\kappa_2},   
  \end{cases} \\
  &  X_\mathpzc{l}(\kappa_1,\kappa_2) = X_\mathpzc{l}(\kappa_1,\kappa_2,0)= X_\mathpzc{l}(\kappa_1,\kappa_2,0,0);\forall \mathpzc{l} \geq 1,\\ &  \hat{X}_\mathpzc{l}(\kappa_1,\kappa_2) = \hat{X}_\mathpzc{l}(\kappa_1,\kappa_2,0)= \hat{X}_\mathpzc{l}(\kappa_1,\kappa_2,0,0);\forall \mathpzc{l} \geq 1,\\ 
& X_\mathpzc{l}(\kappa_1,\kappa_2,\mathfrak{j},i) =  C_{\mathcal{N}} \otimes I_{L_2T^{M_{\mathcal{H}}}_{\kappa_1}T^{M_{\mathcal{N}}}_{\kappa_2}} \otimes P_\mathpzc{l}(\gamma); \forall \mathpzc{l} \geq 0,\kappa_1  = \overline{0,S}, \kappa_2 = S-\kappa_1, \mathfrak{j}=i=0, \\
& \hat{X}_\mathpzc{l}(\kappa_1,\kappa_2,0,0) =   C_{\mathcal{H}} \otimes I_{L_2} \otimes P_{\kappa_1}(\beta_{\mathcal{H}}) \otimes I_{(T^{M_{\mathcal{N}}}_{\kappa_2} \times T^{M_{\mathcal{N}}}_{\kappa_2-1})} \otimes P_\mathpzc{l}(\gamma); \mathpzc{l} \geq 0, \kappa_1  = \overline{0,S-1}~\&~\kappa_1<K_2, \kappa_2 = S-\kappa_1,
\\
&  \mathscr{Q}_{0} = \mathscr{Q}_{M,M-1},\\
& X_{M}(\kappa_1,\kappa_2,\mathfrak{j},i) =  C_{\mathcal{N}} \otimes I_{L_2T^{M_{\mathcal{H}}}_{\kappa_1}T^{M_{\mathcal{N}}}_{\kappa_2}} \otimes P_{M}^{'}(\gamma);\forall \kappa_1  = \overline{0,S}, \kappa_2 = S-\kappa_1, \mathfrak{j}=i=0, \\
& \hat{X}_{M}(\kappa_1,\kappa_2,\mathfrak{j},i) =   C_{\mathcal{H}} \otimes I_{L_2} \otimes P_{\kappa_1}(\beta_{\mathcal{H}}) \otimes I_{(T^{M_{\mathcal{N}}}_{\kappa_2} \times T^{M_{\mathcal{N}}}_{\kappa_2-1})} \otimes P_{M}^{'}(\gamma);  \kappa_1  = \overline{0,S-1},~\&~\kappa_1<K_2,  \kappa_2 = S-\kappa_1, \mathfrak{j}=i=0. \\
& \text{where}~P_{M}^{'}(\gamma) ~\text{is}~ T^{N}_{M} ~\text{order square matrix.}\\
   & \textbf {Lower  Diagonal :}\\
& \mathscr{Q}^{'} =  \mathscr{Q}_{\mathpzc{l},\mathpzc{l}-1} + \mathscr{Q}^{'}_{1,0}, ~~\mathscr{Q}_{\mathpzc{l},\mathpzc{l}-1} = \text{diag}\{Z_{\mathpzc{l}}(0), Z_{\mathpzc{l}}(1), \ldots, Z_{\mathpzc{l}}(S)\}; \mathpzc{l} \geq 1,\\
&Z_{\mathpzc{l}}(\kappa_1) = \text{diag}\{ Z_{\mathpzc{l}}(\kappa_1,0),Z_{\mathpzc{l}}(\kappa_1,1), \ldots, Z_{\mathpzc{l}}(\kappa_1,S-\kappa_1)\}  + \\ & ~~~~~~~~~~~~ \text{diag}^+\{ \hat{Z}_{\mathpzc{l}}(\kappa_1,0), \hat{Z}_{\mathpzc{l}}(\kappa_1,1),\ldots,\hat{Z}_{\mathpzc{l}}(\kappa_1,S-\kappa_1-1)\}; \forall \kappa_1 = \overline{0,S},\\
& Z_1(\kappa_1,\kappa_2) = 
  \begin{cases}
    \textrm{row}(Z_1(0,0,\mathfrak{j}));\forall \mathfrak{j}=\overline{0,S},\\
    \textrm{row}(Z_1(\kappa_1,\kappa_2,0),Z_1(\kappa_1,\kappa_2,S));  \forall S=K  \text{or} ~K<S ~\&~ \kappa_1+\kappa_2 \leq K,\\
      Z_1(\kappa_1,\kappa_2,0);  \text{$\forall K<S ~\&~ \kappa_1+\kappa_2> K$,}
  \end{cases} \\
  &Z_1(\kappa_1,\kappa_2,\mathfrak{j}) = \begin{cases}
Z_1(\kappa_1,\kappa_2,\mathfrak{j},0); 
 \text{$\forall \mathfrak{j} = \overline{0,S-1}$},\\
\textrm{row}(Z_1(\kappa_1,\kappa_2,S,i));\forall i=\overline{0,K-\kappa_1-\kappa_2},
  \end{cases} \\
&   \hat{Z}_1(\kappa_1,\kappa_2) = 
  \begin{cases}
      \textrm{row}(\hat{Z}_1(\kappa_1,\kappa_2,0),\hat{Z}_1(\kappa_1,\kappa_2,S));  \forall S=K \text{or}~K<S ~\&~ \kappa_1+\kappa_2 < K,\\
      \hat{Z}_1(\kappa_1,\kappa_2,0);  \text{$\forall K<S ~\&~ \kappa_1+\kappa_2 \geq K$,}
  \end{cases}\\ & \hat{Z}_1(\kappa_1,\kappa_2,\mathfrak{j}) = \begin{cases}
\hat{Z}_1(\kappa_1,\kappa_2,\mathfrak{j},0); 
\text{$\forall \mathfrak{j} = \overline{0,S-1}$},\\
\textrm{row}(\hat{Z}_1(\kappa_1,\kappa_2,S,i)); \forall i=\overline{0,K-\kappa_1-\kappa_2},   
  \end{cases} \\
  &  Z_\mathpzc{l}(\kappa_1,\kappa_2) = Z_\mathpzc{l}(\kappa_1,\kappa_2,0)= Z_\mathpzc{l}(\kappa_1,\kappa_2,0,0);\mathpzc{l} \geq 2,~~ \hat{Z}_\mathpzc{l}(\kappa_1,\kappa_2) = \hat{Z}_\mathpzc{l}(\kappa_1,\kappa_2,0)= \hat{Z}_\mathpzc{l}(\kappa_1,\kappa_2,0,0);\mathpzc{l} \geq 2,\\
& Z_\mathpzc{l}(\kappa_1,\kappa_2,\mathfrak{j},i) = 
  I_{L_1L_2T_{\kappa_1}^{M_{\mathcal{H}}}T_{\kappa_2}^{M_{\mathcal{N}}}}\otimes L_\mathpzc{l}^{(1)}(\mathpzc{l},\Tilde{\Gamma_1}); \forall \mathpzc{l} \geq 1,\kappa_1  = \overline{0,S}, \kappa_2 = \overline{0,S-\kappa_1}, \mathfrak{j}=i=0, \\
& \hat{Z}_\mathpzc{l}(\kappa_1,\kappa_2,\mathfrak{j},i) =   I_{L_1L_2T_{\kappa_1}^{M_{\mathcal{H}}}}\otimes P_{\kappa_2}(\beta_{\mathcal{N}}) \otimes L_\mathpzc{l}^{(2)}(\mathpzc{l},\Tilde{\Gamma_2}); \forall \mathpzc{l} \geq 1, \kappa_1  = \overline{0,S-1}, \kappa_2 = \overline{0,S-\kappa_1-1}, \mathfrak{j}=i=0, \\
&  \mathscr{Q}_{2} = \mathscr{Q}_{M,M+1}, \\
& Z_M(\kappa_1,\kappa_2,\mathfrak{j},i) = 
  I_{L_1L_2T_{\kappa_1}^{M_{\mathcal{H}}}T_{\kappa_2}^{M_{\mathcal{N}}}}\otimes L_M^{(1)'}(M,\Tilde{\Gamma_1}); \forall \kappa_1  = \overline{0,S}, \kappa_2 = \overline{0,S-\kappa_1}, \mathfrak{j}=i=0, \\
& \hat{Z}_M(\kappa_1,\kappa_2,\mathfrak{j},i) =   I_{L_1L_2T_{\kappa_1}^{M_{\mathcal{H}}}}\otimes P_{\kappa_2}(\beta_{\mathcal{N}}) \otimes L_M^{(2)'}(M,\Tilde{\Gamma_2}); \forall  \kappa_1  = \overline{0,S-1}, \kappa_2 = \overline{0,S-\kappa_1-1}, \mathfrak{j}=i=0, \\
& \text{where}~L_M^{(1)'}(M,\Tilde{\Gamma_1}),~L_M^{(2)'}(M,\Tilde{\Gamma_2}) ~\text{are}~ T^{N}_{M} ~\text{order square matrices.}\\
 & \textbf {Main  Diagonal :}\\
&\mathscr{Q}_{\mathpzc{l},\mathpzc{l}} = \text{diag}\{ Y_{\mathpzc{l}}(0),Y_{\mathpzc{l}}(1), \ldots, Y_{\mathpzc{l}}(S)\} + \text{diag}^+\{ \hat{Y}_{\mathpzc{l}}(0), \hat{Y}_{\mathpzc{l}}(1),\ldots,\hat{Y}_{\mathpzc{l}}(S-1)\} \\ & ~~~~~~~~~+ \text{diag}^-\{ \Bar{Y}_{\mathpzc{l}}(1), \Bar{Y}_{\mathpzc{l}}(2),\ldots,\Bar{Y}_{\mathpzc{l}}(S)\};\forall \mathpzc{l} \geq 0,\\
&Y_{\mathpzc{l}}(\kappa_1) = \text{diag}\{ Y_{\mathpzc{l}}(\kappa_1,0),Y_{\mathpzc{l}}(\kappa_1,1), \ldots, Y_{\mathpzc{l}}(\kappa_1,S-\kappa_1)\} + \text{diag}^-\{ S^N_{\mathpzc{l}}(\kappa_1,1), S^N_{\mathpzc{l}}(\kappa_1,2),\ldots,S^N_{\mathpzc{l}}(\kappa_1,S-\kappa_1)\} \\
&~~~~~~~~~~~~+ \text{diag}^+\{ N_{\mathpzc{l}}(\kappa_1,0), N_{\mathpzc{l}}(\kappa_1,1),\ldots,N_{\mathpzc{l}}(\kappa_1,S-\kappa_1-1)\};\forall \kappa_1 = \overline{0,S},\\
 & Y_0(\kappa_1,\kappa_2) = 
  \begin{cases}
    \text{diag}\{ Y_0(0,0,0),Y_0(0,0,1), \ldots, Y_0(0,0,S)\} 
    + \text{diag}^-\{ R_0(0,0,1), R_0(0,0,2),\ldots,R_0(0,0,S)\},\\
      \text{diag}\{ Y_0(\kappa_1,\kappa_2,0),Y_0(\kappa_1,\kappa_2,S)\};  \text{$\forall S=K$ or $K<S ~\&~ \kappa_1+\kappa_2 \leq K$,}\\
      Y_0(\kappa_1,\kappa_2,0); \text{$\forall K<S ~\&~ \kappa_1+\kappa_2 > K$,}
  \end{cases}\\
  	& Y_0(\kappa_1,\kappa_2,\mathfrak{j}) = \begin{cases}
Y_0(\kappa_1,\kappa_2,\mathfrak{j},0); 
 \text{$\forall \mathfrak{j} = \overline{0,S-1}$},\\
\text{diag}\{ Y_0(\kappa_1,\kappa_2,S,0), \ldots, Y_0(\kappa_1,\kappa_2,S,K-\kappa_1-\kappa_2)\} \\  + \text{diag}^-\{ S^E_0(\kappa_1,\kappa_2,S,1),\ldots, S^E_0(\kappa_1,\kappa_2,S,K-\kappa_1-\kappa_2)\} \\
    + \text{diag}^+\{ E_0(\kappa_1,\kappa_2,S,0), \ldots, E_0(\kappa_1,\kappa_2,S,K-\kappa_1-\kappa_2-1)\}
    \\
   \end{cases} \\
   & N_0(0,0) =
  \begin{pmatrix}
      N_0(0,0,0) & N_0(0,0,1) &  \cdots & N_0(0,0,S-1) & 0\\
       0&0  &  \cdots & 0 & N_0(0,0,S)
  \end{pmatrix}^T,\\
  	 & N_0(\kappa_1,\kappa_2) = 
  \begin{cases}
  \text{diag}\{ N_0(\kappa_1,\kappa_2,0),N_0(\kappa_1,\kappa_2,S)\};  \text{$\forall S=K$ or $K<S ~\&~ \kappa_1+\kappa_2 < K$,} \\
       \textrm{col}( N_0(\kappa_1,\kappa_2,0) , 0)
  ; \text{$\forall  K<S~\&~ \kappa_1+\kappa_2 = K$,}\\
  N_0(\kappa_1,\kappa_2,0);\text{$\forall  K<S~\&~ \kappa_1+\kappa_2 > K$,} 
  \end{cases}\\
  	& N_0(\kappa_1,\kappa_2,\mathfrak{j}) = \begin{cases}
N_0(\kappa_1,\kappa_2,\mathfrak{j},0); 
 \text{$\forall \mathfrak{j} = \overline{0,S-1}$},\\
\begin{pmatrix}
      N_0(\kappa_1,\kappa_2,S,0) &  & \\
  %    0& N_0(\kappa_1,\kappa_2,S,1) & \\
       &  \ddots & \\
      & & N_0(\kappa_1,\kappa_2,S,K-\kappa_1-\kappa_2-1)\\
      &&0
  \end{pmatrix}, 
  \end{cases} \\
   & S^N_0(\kappa_1,\kappa_2) = 
  \begin{cases}
  \begin{pmatrix}
      S^N_0(0,1,0) & S^N_0(0,1,1) &  \cdots& S^N_0(0,1,S-1)&0  \\
      0 &0 &   \cdots &0 &S^N_0(0,1,S)
  \end{pmatrix},\\
      \text{diag}\{ S^N_0(\kappa_1,\kappa_2,0),S^N_0(\kappa_1,\kappa_2,S)\};  \text{$\forall S=K$ or $K<S ~\&~\kappa_1+\kappa_2 \leq K$,}\\
       \textrm{row}(S^N_0(\kappa_1,\kappa_2,0),0);  \text{$\forall K<S~\&~ \kappa_1+\kappa_2 = K+1$,}\\
      S^N_0(\kappa_1,\kappa_2,0);  \text{$\forall K<S~\&~ \kappa_1+\kappa_2 > K+1$,}
  \end{cases}\\
  	& S^N_0(\kappa_1,\kappa_2,\mathfrak{j}) = \begin{cases}
S^N_0(\kappa_1,\kappa_2,\mathfrak{j},0); 
 \text{$\forall \mathfrak{j} = \overline{0,S-1}$},\\
\begin{pmatrix}
      S^N_0(\kappa_1,\kappa_2,S,0) &  & &\\
  %    0& S^N_0(\kappa_1,\kappa_2,S,1) & &\\
       &\ddots &  & \\
      & & S^N_0(\kappa_1,\kappa_2,S,K-\kappa_1-\kappa_2)&  S^N_0(\kappa_1,\kappa_2,S,K-\kappa_1-\kappa_2+1)
  \end{pmatrix};K-\kappa_1-\kappa_2<K_1,\\
  \begin{pmatrix}
      S^N_0(\kappa_1,\kappa_2,S,0) &  & &\\
  %    0& S^N_0(\kappa_1,\kappa_2,S,1) & &\\
       &\ddots &  & \\
      & & S^N_0(\kappa_1,\kappa_2,S,K-\kappa_1-\kappa_2)&  0
  \end{pmatrix};K-\kappa_1-\kappa_2\geq K_1,
  \end{cases} \\
  &  Y_\mathpzc{l}(\kappa_1,\kappa_2) = Y_\mathpzc{l}(\kappa_1,\kappa_2,0)= Y_\mathpzc{l}(\kappa_1,\kappa_2,0,0);\forall \mathpzc{l} \geq 1,\\ & 
     N_\mathpzc{l}(\kappa_1,\kappa_2) = N_\mathpzc{l}(\kappa_1,\kappa_2,0)= N_\mathpzc{l}(\kappa_1,\kappa_2,0,0);\forall \mathpzc{l} \geq 1,\\
   &  S^N_\mathpzc{l}(\kappa_1,\kappa_2) = S^N_\mathpzc{l}(\kappa_1,\kappa_2,0)= S^N_\mathpzc{l}(\kappa_1,\kappa_2,0,0);\forall \mathpzc{l} \geq 1,\\
 &  N_\mathpzc{l}(\kappa_1,\kappa_2,\mathfrak{j},i) =
 \begin{cases}
   C_{\mathcal{N}} \otimes I_{L_2T_{\kappa_1}^{M_{\mathcal{H}}}} \otimes P_{\kappa_2}(\beta_{\mathcal{N}}) \otimes  I_{ T_{\mathpzc{l}}^{N}}; \forall \kappa_1 = \overline{0,S-1}, \kappa_2 = \overline{0,S-\kappa_1-1},  \mathfrak{j} =i=0, \mathpzc{l} \geq 0,\\
C_{\mathcal{N}} \otimes I_{L_2T_{\kappa_1}^{M_{\mathcal{H}}}} \otimes P_{\kappa_2}(\beta_{\mathcal{N}}) \otimes  I_{T^{R}_{\mathfrak{j}}T^{\mathcal{E}}_{i} T_{\mathpzc{l}}^{N}};  \forall \kappa_1 = \overline{0,S-1}, \kappa_2 = \overline{0,S-\kappa_1-1},  \mathfrak{j} =S, \\  i = \overline{0,K-\kappa_1-\kappa_2-1}, \mathpzc{l} = 0,\\
\end{cases}\\
 & S^N_\mathpzc{l}(\kappa_1,\kappa_2,\mathfrak{j},i) =
 \begin{cases}
 I_{L_1L_2T_{\kappa_1}^{M_{\mathcal{H}}}} \otimes L_{S-(\kappa_1+\kappa_2)}(S-\kappa_1, \Tilde{A_{\mathcal{N}}}) \otimes I_{T_{\mathpzc{l}}^{N}}; \forall \kappa_1 = \overline{0,S-1}, \kappa_2 = \overline{1,S-\kappa_1},  \mathfrak{j} =  i =0, \mathpzc{l} \geq 0,\\  
  I_{L_1L_2T_{\kappa_1}^{M_{\mathcal{H}}}} \otimes L_{K-(\kappa_1+\kappa_2+i)}(K-\kappa_1-i, \Tilde{A_{\mathcal{N}}}) \otimes I_{T^{\mathcal{E}}_{i}T^{R}_{S} }; \forall \kappa_1 = \overline{0,S-1}, \kappa_2 = \overline{1,S-\kappa_1}, \mathfrak{j} = S, \\ i = \overline{0,K-\kappa_1-\kappa_2},\mathpzc{l} = 0,\\
  C_{\mathcal{E}} \otimes I_{L_2T_{\kappa_1}^{M_{\mathcal{H}}}(T_{\kappa_2}^{M_{\mathcal{N}}}\times T_{\kappa_2-1}^{M_{\mathcal{N}}})} \otimes P_i(\beta_\mathcal{E}) \otimes I_{T_{S}^{R}};\forall \kappa_1 = \overline{0,S-1}, \kappa_2 = \overline{1,S-\kappa_1}, \mathfrak{j} = S, \\ i = K-\kappa_1-\kappa_2+1,\mathpzc{l} = 0,
    \end{cases}\\
   &E_0(\kappa_1,\kappa_2,S,i) = C_{\mathcal{E}} \otimes I_{L_2T_{\kappa_1}^{M_{\mathcal{H}}}T_{\kappa_2}^{M_{\mathcal{N}}}} \otimes P_{i}(\beta_{\mathcal{E}})\otimes  I_{T^{R}_{\mathfrak{j}}}; \forall \kappa_1 = \overline{0,S}, \kappa_2 = \overline{0,S-\kappa_1},  i = \overline{0,K-\kappa_1-\kappa_2-1},\\
   & S^E_0(\kappa_1,\kappa_2,S,i) = I_{L_1L_2T_{\kappa_1}^{M_{\mathcal{H}}}T_{\kappa_2}^{M_{\mathcal{N}}}} \otimes L_{K-(\kappa_1+\kappa_2+i)}(K-\kappa_1-\kappa_2, \Tilde{A_{\mathcal{E}}}) \otimes I_{T^{R}_S}; \forall \kappa_1 = \overline{0,S}, \kappa_2 = \overline{0,S-\kappa_1}, \\ &~~~~~~~~~~~~~~~~~~~~~~~~~ i = \overline{1,K-\kappa_1-\kappa_2},\\
   & R_0(0,0,\mathfrak{j})=R_0(0,0,\mathfrak{j},0)=  I_{L_1L_2} \otimes L_{S-\mathfrak{j}}(S, \Tilde{B});  \mathfrak{j} = \overline{1,S},\\
   & Y_\mathpzc{l}(\kappa_1,\kappa_2,\mathfrak{j},i) = 
 \begin{cases}
  & C_0 \oplus D(1)+ \Delta;\kappa_1=\kappa_2=\mathfrak{j}=i=\mathpzc{l}=0,\\ &  (C_0+C_{\mathcal{N}}+C_{\mathcal{H}}) \oplus D(1) \oplus A_{\mathfrak{j}}(S,B) + \Delta;\kappa_1=\kappa_2=i=\mathpzc{l}=0,\mathfrak{j}= \overline{1,S-1},\\
   & C_0 \oplus D(1)\oplus A_{\kappa_1}(S-\kappa_2,A_{\mathcal{H}}) \oplus A_{\kappa_2}(S-\kappa_1,A_{\mathcal{N}}) \oplus A_{i}(K-\kappa_1-\kappa_2,A_{\mathcal{E}})\oplus A_{\mathfrak{j}}(S,B)+ \Delta; \\ & \kappa_1=\overline{0,S}, \kappa_2 = \overline{0,S-\kappa_1} \mathpzc{l}=0, \mathfrak{j}=S,i= \overline{1,K-\kappa_1-\kappa_2},\\
 & (C_0+C_{\mathcal{N}}+C_{\mathcal{H}}+C_{\mathcal{E}}) \oplus D(1)  \oplus  A_{i}(K-\kappa_1-\kappa_2,A_{\mathcal{E}})\oplus A_{\mathfrak{j}}(S,B)+ \Delta;  \\ & \kappa_1=0, \kappa_2 =0, \mathpzc{l}=0, \mathfrak{j}=S,i= K-\kappa_1-\kappa_2,\\
 & (C_0+C_{\mathcal{N}}+C_{\mathcal{H}}) \oplus D(1) \oplus A_{\kappa_1}(K-\kappa_2-i,A_{\mathcal{H}}) \oplus A_{\kappa_2}(K-\kappa_1-i,A_{\mathcal{N}}) \oplus A_{i}(K-\kappa_1-\kappa_2,A_{\mathcal{E}}) \\ & \oplus A_{\mathfrak{j}}(S,B)+ \Delta;   \kappa_1=\overline{0,S}, \kappa_2 = \overline{0,S-\kappa_1}, \mathpzc{l}=0, \mathfrak{j}=S,i= \overline{1,K-\kappa_1-\kappa_2},\\
& C_0 \oplus D_0 \oplus A_{\kappa_1}(S-\kappa_2,A_{\mathcal{H}}) \oplus A_{\kappa_2}(S-\kappa_1,A_{\mathcal{N}}) \oplus  A_{\mathpzc{l}}(\mathpzc{l},\Gamma)+ \Delta;\\
& \kappa_1=\overline{0,S}, \kappa_2 = \overline{0,S-\kappa_1-1}, \mathpzc{l} \geq 0, \mathfrak{j}=i= 0,\\
& (C_0+C_{\mathcal{H}}) \oplus D_0 \oplus A_{\kappa_1}(S-\kappa_2,A_{\mathcal{H}}) \oplus A_{\kappa_2}(S-\kappa_1,A_{\mathcal{N}}) \oplus  A_{\mathpzc{l}}(\mathpzc{l},\Gamma)\\ & +  I_{L_2T_{\kappa_1}^{M_{\mathcal{H}}}T_{\kappa_2}^{M_{\mathcal{N}}}} \otimes L_{\mathpzc{l}}^{(2)}(\mathpzc{l},\Tilde{\Gamma_2})P_{\mathpzc{l}}(\gamma)+ \Delta^{'};   \kappa_1=\overline{0,S}, \kappa_2 = S-\kappa_1, \mathpzc{l} \geq 0, \mathfrak{j}=i= 0,
 \end{cases}\\
 & \Delta = \text{diag}\{I_{L_1L_2} \otimes \Delta^{(\kappa_1, \kappa_2, \mathfrak{j}, i)} \}, \Delta^{'} = \text{diag}\{I_{L_1L_2} \otimes \Delta^{(\kappa_1, \kappa_2, \mathfrak{j}, i)^{'}} \}, \\ & \Delta^{(\kappa_1, \kappa_2, \mathfrak{j}, i)^{'}} = \Delta^{(\kappa_1, \kappa_2, \mathfrak{j}, i)} -\text{diag}\{ [I_{L_2T_{\kappa_1}^{M_{\mathcal{H}}}T_{\kappa_2}^{M_{\mathcal{N}}}} \otimes L_{\mathpzc{l}}^{(2)}(\mathpzc{l},\Tilde{\Gamma_2})P_{\mathpzc{l}}(\gamma)]e\}, \\ &\Delta^{(\kappa_1, \kappa_2, \mathfrak{j}, i)} = -\text{diag}\{ [A_{\kappa_1}(S-\kappa_2,A_{\mathcal{H}}) \oplus A_{\kappa_2}(S-\kappa_1,A_{\mathcal{N}}) \oplus  A_{i}(K-\kappa_1-\kappa_2,A_{\mathcal{E}}) \oplus A_{\mathfrak{j}}(S,B) \oplus A_{\mathpzc{l}}(\mathpzc{l},\Gamma)]e\}, \\
    & \hat{Y}_{\mathpzc{l}}(\kappa_1)=
   \begin{pmatrix}
      \hat{Y}_{\mathpzc{l}}(\kappa_1,0) &  & \\
    & \hat{Y}_{\mathpzc{l}}(\kappa_1,1) & \\
       &  &  \\
       &\ddots&\\
       &&\\
      & & \hat{Y}_{\mathpzc{l}}(\kappa_1,S-\kappa_1-1)\\
      &&0
  \end{pmatrix},\\
  & \hat{Y}_0(\kappa_1,\kappa_2) = 
  \begin{cases}
 \begin{pmatrix}
      \hat{Y}_0(0,0,0)  &\cdots & \hat{Y}_0(0,0,S-1)&0\\
      0  &   \cdots &0   &  \hat{Y}_0(0,0,S)
  \end{pmatrix}^T,\\
      \text{diag}\{ \hat{Y}_0(\kappa_1,\kappa_2,0),\hat{Y}_0(\kappa_1,\kappa_2,S)\};  \forall S=K  ~\text{or}~ K<S~\&~ \kappa_1+\kappa_2 < K,\\
     \textrm{col}(\hat{Y}_0(\kappa_1,\kappa_2,0),\hat{Y}_0(\kappa_1,\kappa_2,S)); \forall S=K  ~\text{or} ~~K<S~\&~ \kappa_1+\kappa_2 = K,\\
     \hat{Y}_0(\kappa_1,\kappa_2,0);  \forall S=K ~ \text{or}~ K<S~\&~ \kappa_1+\kappa_2 > K,
     \end{cases}\\
   	& \hat{Y}_0(\kappa_1,\kappa_2,\mathfrak{j}) = \begin{cases}
\hat{Y}_0(\kappa_1,\kappa_2,\mathfrak{j},0); 
 \text{$\forall \mathfrak{j} = \overline{0,S-1}$},\\
\begin{pmatrix}
      \hat{Y}_0(\kappa_1,\kappa_2,S,0) &  & \\
  %    0& \hat{Y}_0(\kappa_1,\kappa_2,S,1) & \\
       &  \ddots&  \\
      & & \hat{Y}_0(\kappa_1,\kappa_2,S,K-\kappa_1-\kappa_2-1)\\
      &&0
  \end{pmatrix},
  \end{cases} \\
   &  \hat{Y}_\mathpzc{l}(\kappa_1,\kappa_2) = \hat{Y}_\mathpzc{l}(\kappa_1,\kappa_2,0)= \hat{Y}_\mathpzc{l}(\kappa_1,\kappa_2,0,0);\mathpzc{l} \geq 1,\\
  & \hat{Y}_\mathpzc{l}(\kappa_1,\kappa_2,\mathfrak{j},i) =
  \begin{cases}
    C_{\mathcal{H}} \otimes I_{L_2} \otimes P_{\kappa_1}(\beta_{\mathcal{H}}) \otimes I_{T_{\kappa_2}^{M_{\mathcal{N}}} T_{\mathpzc{l}}^{N}};  \kappa_1 = \overline{0,S-1}, \kappa_2 = \overline{0,S-\kappa_1-1}, \mathfrak{j} = i=0, \mathpzc{l} \geq 1,\\
    C_{\mathcal{H}} \otimes I_{L_2} \otimes P_{\kappa_1}(\beta_{\mathcal{H}}) \otimes I_{T_{\kappa_2}^{M_{\mathcal{N}}}T^{\mathcal{E}}_{i}T^{R}_{\mathfrak{j}} T_{\mathpzc{l}}^{N}};  \kappa_1 = \overline{0,S-1}, \kappa_2 = \overline{0,S-\kappa_1-1}, \mathfrak{j} = S, i = \overline{0,K-\kappa_1-\kappa_2}, \mathpzc{l} = 0,\\
  \end{cases}\\
  	 & \Bar{Y}_{\mathpzc{l}}(\kappa_1)=
   \begin{pmatrix}
      \Bar{Y}_{\mathpzc{l}}(\kappa_1,0) &  & &\\
     & \Bar{Y}_{\mathpzc{l}}(\kappa_1,1) & &\\
       & \ddots &  & \\
      & & \Bar{Y}_{\mathpzc{l}}(\kappa_1,S-\kappa_1) & 0
  \end{pmatrix},\\ &
    \Bar{Y}_0(\kappa_1,\kappa_2) = 
  \begin{cases}
 \begin{pmatrix}
      \Bar{Y}_0(1,0,0)   & \cdots& \Bar{Y}_0(1,0,S-1)&0 \\
       0&  \cdots&0& \Bar{Y}_0(1,0,S)
  \end{pmatrix},\\
      \text{diag}\{ \Bar{Y}_0(\kappa_1,\kappa_2,0),\Bar{Y}_0(\kappa_1,\kappa_2,S)\};  \text{$\forall S=K$ or $K<S~\&~ \kappa_1+\kappa_2 \leq K$,}\\
     \textrm{row}( \Bar{Y}_0(\kappa_1,\kappa_2,0),\Bar{Y}_0(\kappa_1,\kappa_2,S));  \text{$\forall S=K$ or $K<S ~\&~ \kappa_1+\kappa_2 = K+1$,}\\
      \Bar{Y}_0(\kappa_1,\kappa_2,0); \text{$\forall K<S~\&~ \kappa_1+\kappa_2 > K+1$,}
  \end{cases}\\
   	& \Bar{Y}_0(\kappa_1,\kappa_2,\mathfrak{j}) = \begin{cases}
\Bar{Y}_0(\kappa_1,\kappa_2,\mathfrak{j},0); ~~~ \text{$\forall \mathfrak{j} = \overline{0,S-1}$},\\
\begin{pmatrix}
      \Bar{Y}_0(\kappa_1,\kappa_2,S,0) &  & &\\
%      0& \Bar{Y}_0(\kappa_1,\kappa_2,S,1) & &\\
       & \ddots &  & \\
      & & \Bar{Y}_0(\kappa_1,\kappa_2,S,K-\kappa_1-\kappa_2-1)&  \Bar{Y}_0(\kappa_1,\kappa_2,S,K-\kappa_1-\kappa_2)
  \end{pmatrix};\\ \kappa_1\leq K,\kappa_2=0, K-\kappa_1-\kappa_2-1<K_1, \\
  \begin{pmatrix}
      \Bar{Y}_0(\kappa_1,\kappa_2,S,0) &  & &\\
%      0& \Bar{Y}_0(\kappa_1,\kappa_2,S,1) & &\\
       & \ddots &  & \\
      & & \Bar{Y}_0(\kappa_1,\kappa_2,S,K-\kappa_1-\kappa_2-1)&  0
  \end{pmatrix};K-\kappa_1-\kappa_2-1\geq K_1,  \\
  \end{cases} \\
  	 &  \Bar{Y}_\mathpzc{l}(\kappa_1,\kappa_2) = \Bar{Y}_\mathpzc{l}(\kappa_1,\kappa_2,0)= \Bar{Y}_\mathpzc{l}(\kappa_1,\kappa_2,0,0);\forall \mathpzc{l} \geq 1,\\
   & \Bar{Y}_\mathpzc{l}(\kappa_1,\kappa_2,0,0) =  
      I_{L_1L_2} \otimes L_{S-(\kappa_1+\kappa_2)}(S-\kappa_2, \Tilde{A}_{\mathcal{H}}) \otimes I_{T_{\kappa_2}^{M_{\mathcal{N}}}T_{\mathpzc{l}}^{N}}; \forall \kappa_1 = \overline{1,S}, \kappa_2 = \overline{0,S-\kappa_1},   \mathpzc{l} \geq 0,\\
   & \Bar{Y}_0(\kappa_1,\kappa_2,S,i) = I_{L_1L_2} \otimes L_{K-(\kappa_1+\kappa_2+i)}(K-\kappa_2-i, \Tilde{A}_{\mathcal{H}}) \otimes I_{T_{\kappa_2}^{M_{\mathcal{N}}}T^{\mathcal{E}}_{i}T^{R}_{S}};\forall  \kappa_1 = \overline{1,S}, \kappa_2 = \overline{0,S-\kappa_1},  i = \overline{0,K-\kappa_1-\kappa_2},\\
   & \Bar{Y}_0(\kappa_1,0,S,K-\kappa_1-\kappa_2) = C_{\mathcal{E}}\otimes I_{L_2(T_{\kappa_1}^{M_{\mathcal{H}}}\times T_{\kappa_1-1}^{M_{\mathcal{H}}})} \otimes P_i(\beta_{\mathcal{E}}) \otimes I_{T_{S}^{R}};\forall \kappa_1 = \overline{1,S}, i=K-\kappa_1.\\
   & \textbf {First column:}\\
    &\mathscr{Q}^{'}_{\mathpzc{l},0} = 
    \begin{pmatrix}
        W_\mathpzc{l}(0)&0&\cdots &0\\
        W_\mathpzc{l}(1)&0&\cdots &0\\
        \vdots&\vdots&\ddots &\vdots \\
        W_\mathpzc{l}(S)&0&\cdots &0
    \end{pmatrix};\forall \mathpzc{l}\geq 0, ~~~~~ W_\mathpzc{l}(\kappa_1) = 
    \begin{pmatrix}
        W_\mathpzc{l}(\kappa_1,0)&0&\cdots &0\\
        W_\mathpzc{l}(\kappa_1,1)&0&\cdots &0\\
        \vdots&\vdots&\ddots& \vdots \\
        W_\mathpzc{l}(\kappa_1,S-\kappa_1)&0&\cdots &0
    \end{pmatrix}; \forall \kappa_1 = \overline{0,S},\\
    & \text{where $\mathscr{Q}^{'}_{\mathpzc{l},0}$ and $ W_\mathpzc{l}(\kappa_1)$ are square matrices of order $S+1$ and $S-\kappa_1+1$, respectively,}\\ & \mathscr{Q}^{+} = \mathscr{Q}_{M,0},\\
    & W_0(0,0) = \text{diag}\{W_0(0,0,0),W_0(0,0,1),\ldots,W_0(0,0,S)\},\\ 
    & W_0(\kappa_1,\kappa_2) = 
    \begin{cases}
      \begin{pmatrix}
        W_0(\kappa_1,\kappa_2,0) & W_0(\kappa_1,\kappa_2,1) & \cdots & W_0(\kappa_1,\kappa_2,S-1)\\
      0  & 0 &\cdots & 0 & W_0(\kappa_1,\kappa_2,S) \end{pmatrix}; \\\text{$\forall S=K$ or $K<S~\&~ \kappa_1+\kappa_2 \leq K$,}\\
      \textrm{row}(W_\mathpzc{l}(\kappa_1,\kappa_2,0),W_\mathpzc{l}(\kappa_1,\kappa_2,1),\ldots,W_\mathpzc{l}(\kappa_1,\kappa_2,S));\text{$\forall  K<S~\&~ \kappa_1+\kappa_2 >K$},
    \end{cases}\\
  &   W_0(\kappa_1,\kappa_2,\mathfrak{j}) = 
     \begin{cases}
      W_0(\kappa_1,\kappa_2,\mathfrak{j},0); \forall \mathfrak{j} = \overline{0,S-1},\\
       \text{diag}\{W_0(0,0,S,0),W_0(0,0,S,1),\ldots,W_0(0,0,S,K)\},\\
       \begin{pmatrix}
        W_0(\kappa_1,\kappa_2,S,0) &  &  & \\
  %       & W_0(\kappa_1,\kappa_2,S,1)&  &  & \\
         && \ddots &&\\
         &&& W_0(\kappa_1,\kappa_2,S,K-\kappa_1-\kappa_2-1) & 0
        \end{pmatrix}; \text{$\forall \kappa_1+\kappa_2 < K$,}\\
        \textrm{row}\{W_0(\kappa_1,\kappa_2,S,0),W_0(\kappa_1,\kappa_2,S,1),\ldots,W_0(\kappa_1,\kappa_2,S,K-\kappa_1-\kappa_2)\};\text{$\forall  \kappa_1+\kappa_2 \geq K$},
     \end{cases}\\
    & W_\mathpzc{l}(\kappa_1,\kappa_2,\mathfrak{j},0) = I_{L_1} \otimes D_1 \otimes e_{T^{M_{\mathcal{H}}}_{\kappa_1}T^{M_{\mathcal{N}}}_{\kappa_2}T^{M_{\mathcal{E}}}_{i}} \otimes  \Pi P_{\mathfrak{j}}(\alpha) \otimes e_{T^{N}_{\mathpzc{l}}};\forall \kappa_1 = \overline{0,S}, \kappa_2 = \overline{0,S-\kappa_1}, \mathfrak{j} = \overline{1,S}~ \& ~ \mathfrak{j}=\kappa_1+\kappa_2, \mathpzc{l} \geq 0.
\end{align*}}}

\subsection{Ergodicity Condition}
The modeled process  \{$\Xi(t), t \geq 0 \}$ clearly has the traits of a level-dependent quasi-birth-death ($\textit{L\!D\!Q\!B\!D}$) process as it is observed by its structure. But, the existing definition of $\textit{L\!D\!Q\!B\!D}$ process does not provide any results about the limiting/asymptotic behaviour of process when the countable element of Markov chain tends to infinity. 
The behaviour of the proposed process satisfies the conditions imposed on the limiting behaviour of asymptotic quasi-Toeplitz Markov chain ($\textrm{\it A\!Q\!T\!M\!C}$). To prove this, we will compute matrices $U_0, U_1$ and $U_2$ (refer, \cite{klimenok2006multi}) defined as follows:\\
\[{\displaystyle
U_0 = \lim_{\mathpzc{l} \to \infty} T_{\mathpzc{l}}^{-1} \mathscr{Q}_{\mathpzc{l},\mathpzc{l}-1}, U_1 = \lim_{\mathpzc{l} \to \infty} T_{\mathpzc{l}}^{-1} \mathscr{Q}_{\mathpzc{l},\mathpzc{l}} + I, U_2 = \lim_{\mathpzc{l} \to \infty} T_{\mathpzc{l}}^{-1} \mathscr{Q}_{\mathpzc{l},\mathpzc{l}+1}},\]\\
where $T_{\mathpzc{l}}^{-1}$ is the diagonal matrix with diagonal entries defined as the modulus of the diagonal entries of the matrix $\mathscr{Q}_{\mathpzc{l},\mathpzc{l}}, \mathpzc{l} \geq 0.$ The matrices $U_0, U_1$ and $U_2$ have the following form\\
\[{\displaystyle
U_0 = \lim_{\mathpzc{l} \to \infty} T_{\mathpzc{l}}^{-1} \mathscr{Q}_{\mathpzc{l},\mathpzc{l}-1}= T^{-1}\mathscr{Q}_0; \mathpzc{l}>M,\\}\]
\[{\displaystyle
U_1 = \lim_{\mathpzc{l} \to \infty} T_{\mathpzc{l}}^{-1} \mathscr{Q}_{\mathpzc{l},\mathpzc{l}} + I= T^{-1}\mathscr{Q}_1;\mathpzc{l}>M,\\}\]
\[{\displaystyle
U_2 = \lim_{\mathpzc{l} \to \infty} T_{\mathpzc{l}}^{-1} \mathscr{Q}_{\mathpzc{l},\mathpzc{l}+1}= T^{-1}\mathscr{Q}_2;\mathpzc{l}>M,\\}\]
where {\small{
\begin{align*}
& T = \text{diag}\{ T(0), T(1),\ldots, T(S)\}; T(\kappa_1) = \text{diag}\{ T(\kappa_1,0), T(\kappa_1,1),\ldots, T(\kappa_1,S-\kappa_1)\},\\
&T(\kappa_1,\kappa_2)=
\begin{cases}
  \Lambda_0 \oplus \sum \oplus A_{\kappa_1}^{'}(S-\kappa_2, A_{\mathcal{H}}) \oplus A_{\kappa_2}^{'}(S-\kappa_1, A_{\mathcal{N}}) \oplus A_{M}^{'}(M, \gamma)
  ; \kappa_1=\overline{0,S}, \kappa_2= \overline{0,S-\kappa_1-1},\\
     \Lambda_0+C_{\mathcal{H}} \oplus \sum \oplus A_{\kappa_1}^{'}(S-\kappa_2, A_{\mathcal{H}}) \oplus A_{\kappa_2}^{'}(S-\kappa_1, A_{\mathcal{N}}) \oplus A_{M}^{'}(M, \gamma)\\
     + I_{L_1L_2T_{\kappa_1}^{M_{\mathcal{H}}}T_{\kappa_2}^{M_{\mathcal{N}}}}\otimes L_0^{(2)}(M,\Tilde{\Gamma_2})\otimes P_M(\gamma)
  ; \kappa_1=\overline{0,S}, \kappa_2= S-\kappa_1.\\
  \end{cases}
    \end{align*}}}
    
     Here, $\Lambda_0 $ and $\sum $ are diagonal matrices whose diagonal entries are defined by the diagonal entries of the matrices $-C_0$ and $-D_0 +D_1,$ respectively. $A_{\kappa_1}^{'}(S-\kappa_2, A_{\mathcal{H}}), A_{\kappa_2}^{'}(S-\kappa_1, A_{\mathcal{N}})$, and $A_{M}^{'}(M, \gamma)$ are diagonal matrices whose diagonal entries are defined by the diagonal entries of the matrices $A_{\kappa_1}(S-\kappa_2, A_{\mathcal{H}}), A_{\kappa_2}(S-\kappa_1, A_{\mathcal{N}})$, and $A_{M}(M, \gamma)$, respectively. The existence of limiting matrices $U_0, U_1,$ and $ U_2$ proves that the Markov chain \{$\Xi(t), t \geq 0 \}$ belongs to the class of $\textrm{\it A\!Q\!T\!M\!C}$. Further, this fact has been used to prove the ergodicity condition of the underlying process.

\begin{theorem}
The necessary and sufficient condition for the ergodicity of the underlying process is the satisfaction of the inequality \\
\begin{align*}
\lambda <  \sum_{\kappa_1=0}^{S} \sum_{\kappa_2=0}^{S-\kappa_1} x_M^{(1)} L_0^{(1)'}(M,\Tilde{\Gamma_1})e  + \sum_{\kappa_1=0}^{S} \sum_{\kappa_2=0}^{S-\kappa_1} x_M^{(2)} L_0^{(2)'}(M,\Tilde{\Gamma_2})e
\end{align*}

 \noindent where $x_M^{(1)} = x_M(\kappa_1,\kappa_2)(e_{L_2T_{\kappa_1}^{M_{\mathcal{H}}}T_{\kappa_2}^{M_{\mathcal{N}}}}\otimes I_{T^N_M})$, $x_M^{(2)} = x_M(\kappa_1,\kappa_2)(e_{L_2T_{\kappa_1}^{M_{\mathcal{H}}}T_{\kappa_2}^{M_{\mathcal{N}}}}\otimes I_{T^N_M})$ and $x$ is the unique solution of \begin{center}
 $x(\mathscr{Q}_0+\mathscr{Q}_1+\mathscr{Q}_2)=0; xe=1.$    
\end{center}
\end{theorem}
\begin{proof}
Following \cite{klimenok2006multi}, a necessary and sufficient condition for ergodicity of the underlying process can be formulated in terms of the generator $\mathscr{Q}$ is as follows
\begin{align}
     x\mathscr{Q}_2e < x\mathscr{Q}_0e, \label{eq:1}
\end{align}
   where  $x$ is the unique solution to the system of linear equations 
   \begin{align}
        x(\mathscr{Q}_0+\mathscr{Q}_1+\mathscr{Q}_2)=0; xe=1 \label{eq:2}.
   \end{align}
  Let expression for $x$ be of the form
  $x = (\pi \otimes x_M(\kappa_1,0),\pi \otimes x_M(\kappa_1,1), \ldots, \pi \otimes x_M(\kappa_1,S-\kappa_1))$, where $x_M(\kappa_1,\kappa_2)$  is of size $L_2T^{M_{\mathcal{H}}}_{\kappa_1}T^{M_{\mathcal{N}}}_{\kappa_2}T^N_M.$ Substituting the expression of $x$ and block matrices $\mathscr{Q}_0$ and $\mathscr{Q}_2$ from the description of generator matrix in (\ref{eq:1}). 
  \begin{align}
  \nonumber    & \sum_{\kappa_1=0}^{S}(\pi \otimes x_M(\kappa_1,S-\kappa_1)) (C_{\mathcal{N}} \otimes I_{L_2T_{\kappa_1}^{M_{\mathcal{H}}}T_{\kappa_2}^{M_{\mathcal{N}}}}\otimes P_M^{'}(\gamma) )e\\ 
  \nonumber     & + \sum_{\kappa_1=1}^{S}  (\pi \otimes x_M(\kappa_1,S-\kappa_1)) (C_{\mathcal{H}} \otimes I_{L_2} \otimes P_{\kappa_1}(\beta_{\mathcal{H}}) \otimes I_{T_{\kappa_2}^{M_{\mathcal{N}}}\times 
      T_{\kappa_2-1}^{M_{\mathcal{N}}}}\otimes P_M^{'}(\gamma) )e\\
  \nonumber     & < \sum_{\kappa_1=0}^{S} \sum_{\kappa_2=0}^{S-\kappa_1} (\pi \otimes x_M(\kappa_1,\kappa_2)) ( I_{L_1L_2T_{\kappa_1}^{M_{\mathcal{H}}}T_{\kappa_2}^{M_{\mathcal{N}}}}\otimes L_0^{(1)'}(M,\Tilde{\Gamma_1}) )e\\ 
      & + \sum_{\kappa_1=0}^{S} \sum_{\kappa_2=1}^{S-\kappa_1} (\pi \otimes x_M(\kappa_1,\kappa_2)) ( I_{L_1L_2T_{\kappa_1}^{M_{\mathcal{H}}}} \otimes P_{\kappa_2}(\beta_{\mathcal{N}})\otimes L_0^{(2)'}(M,\Tilde{\Gamma_2}) )e, \label{eq:3}
  \end{align}
 using the relations $\pi C_{\mathcal{N}}e = \lambda_{\mathcal{N}},\pi C_{\mathcal{H}}e = \lambda_{\mathcal{H}}, \lambda_{\mathcal{N}}+\lambda_{\mathcal{H}}=\lambda, P_M^{'}(\gamma)e=e, P_{\kappa_1}(\beta_{\mathcal{H}})e=e, P_{\kappa_2}(\beta_{\mathcal{N}})e=e$, where $e$ is a column vector of appropriate size  with all the entries one, the inequality (\ref{eq:3}) will be reduced to 

\begin{align*}
\lambda <  \sum_{\kappa_1=0}^{S} \sum_{\kappa_2=0}^{S-\kappa_1} x_M^{(1)} L_0^{(1)'}(M,\Tilde{\Gamma_1})e  + \sum_{\kappa_1=0}^{S} \sum_{\kappa_2=1}^{S-\kappa_1} x_M^{(2)} L_0^{(2)'}(M,\Tilde{\Gamma_2})e.
\end{align*}
Since the stability of the system can not be defined in the catastrophic scenario, thus these conditions are derived for the normal scenario.
System of equations (\ref{eq:2}) has unique solution because the matrix of the system is an infinitesimal generator of the underlying process which defines joint distributions of the number of retrial calls, number of handoff calls receiving service and number of new calls receiving service. The left hand side of  (\ref{eq:1})  is the total arrival rate of handoff calls and new calls in the system. In the right hand side, the first summand is the rate of  departure from the system  and second summand is the rate of starting the service for retrial calls when the retrial is successful. It is clear that the Markov chain describing queueing model under study is ergodic if and only if the total arrival rate is less than the maximum value of the total departure rate and successful retrial rate. When the number of retrial calls increases without bound, i.e.,  the retrial rate tends to infinity, the retrial queueing model approaches to the corresponding classical queueing model for which $\lambda/S\mu$ becomes a necessary and sufficient condition for the stability.
\end{proof}

\subsection{Stationary Distribution}
Let $\displaystyle{{z_s} = \{{z_s}(0), {z_s}(1),{z_s}(2),\ldots, {z_s}(M-1), {z_s}(M), \ldots \}}$ be the steady-state probability vector of generator matrix $\mathscr{Q}$ satisfying
\begin{center}

$\displaystyle{{z_s} \mathscr{Q} = 0; {z_s} e =1.}$ 
    
\end{center}
Here, element ${z_s}(0)$ contains   $\displaystyle{1 \times  \Big(\sum_{\mathfrak{j}=0}^{S} \sum_{i=0}^{K} L_1L_2 T^{M_{\mathcal{E}}}_{i}T^R_{\mathfrak{j}}\Big) \Big(\sum_{\kappa_1=0}^{S} \sum_{\kappa_2=0}^{S}\sum_{\mathfrak{j}=0}^{S}  L_1L_2 T^{M_{\mathcal{H}}}_{\kappa_1} T^{M_{\mathcal{N}}}_{\kappa_2}T^R_{S}\Big)} $ vector components and ${z_s}(\mathpzc{l})$ contains 
 $\displaystyle{1 \times \Big(\sum_{\kappa_1=0}^{S} \sum_{\kappa_2=0}^{S}  L_1L_2 T^{M_{\mathcal{H}}}_{\kappa_1} T^{M_{\mathcal{N}}}_{\kappa_2}T^N_{\mathpzc{l}}\Big)} $ elements; $  \mathpzc{l} \geq 0,~ 0\leq \kappa_1 \leq S,~ 0\leq \kappa_2 \leq S.$
The derived structure of generator matrix lacks of the existing quasi-birth death structure and Toeplitz like structure. Therefore, the existing approach for computing the stationary distribution for $\textrm{\it A\!Q\!T\!M\!C}$ can be employed because the Markov chain has a specific asymptotic behaviour. Dudin and Dudina \cite{dudin2019retrial}  proposed an  algorithm for $\textrm{\it A\!Q\!T\!M\!C}$ process. In their proposed approach, they dealt with the challenges of large order matrix computation and storage. Along with the substantial advantages, that algorithm also has the following disadvantages.
\begin{itemize}
    \item In their algorithm, %the initial value for starting the computation process has been chosen  randomly. 
    they did not provide any reasoning on how to select the initial point $i_0.$
    \item After each unsuccessful iteration, to obtain a new search interval, a randomly selected value $s$ has been added in the existing interval. There is no explanation behind the chosen value of $s$.
    \item While checking the termination criteria for stationary distribution $z_s,$ they discarded some portion of the search interval at each failed iteration of Step 3.2 (Case 3). They did not provide any clarification for the eliminated portion of the search interval.
    
\end{itemize}

 These findings motivate for the development of a modified approach for computing the stationary distribution of the Markov chain under consideration.

\subsubsection{Old Algorithm}
\noindent \textbf{Step 1.} Set $i_0$ and $s$ randomly. Fix $\epsilon_g$ and $\epsilon_f$ as accuracy levels of matrices $G_i$ and steady-state vector $z_s(i),$ respectively. 

\noindent \textbf{Step 2.} Compute  matrix $G$.

\noindent \textbf{Step 2.1.}  Set $G_{\kappa}^{(1)}=O$ and $G_{\kappa}^{(2)}=I$, $\kappa=i_0-1+s,$ $i_f=-1.$ 

\noindent \textbf{Step 2.2.} Compute  matrices  $G_{\kappa}^{(1)}$ and $G_{\kappa}^{(2)}$ defined as:
\begin{center}

$G_{\kappa}^{(1)}:= -(\mathscr{Q}_{{\kappa}+1,{\kappa}+1} +\mathscr{Q}_{{\kappa}+1,{\kappa}+2}  G_{{\kappa}+1}^{(1)})^{-1}\mathscr{Q}_{{\kappa}+1,{\kappa}}$,  \\
$G_{\kappa}^{(2)}:= -(\mathscr{Q}_{{\kappa}+1,{\kappa}+1} +\mathscr{Q}_{{\kappa}+1,{\kappa}+2}  G_{{\kappa}+1}^{(2)})^{-1}\mathscr{Q}_{{\kappa}+1,{\kappa}}$.
\end{center}

\noindent \textbf{Step 2.3.} Calculate $||G_{\kappa}^{(1)}-G_{\kappa}^{(2)}||.$ There can be three possible cases as follows. 

\textit{Case 1:} If $||G_{\kappa}^{(1)}-G_{\kappa}^{(2)}|| < \epsilon_g$, go to Step 2.4. 

\textit{Case 2:} If $||G_{\kappa}^{(1)}-G_{\kappa}^{(2)}|| > \epsilon_g$ and ${\kappa} \geq i_0$, set ${\kappa} = {\kappa}-1$ and repeat Step 2.2.

\textit{Case 3:} If $||G_{\kappa}^{(1)}-G_{\kappa}^{(2)}|| > \epsilon_g$ and ${\kappa} = i_0-1$, set $s=2s,$ ${\kappa} =i_0-1+2s,$ $i_0={\kappa},$ and go to  Step 2.1.

\noindent \textbf{Step 2.4.} Set $G_{\kappa} = G_{\kappa}^{(1)}.$ Compute  matrices $G_i, i=(i_f+1,{\kappa}-1)$. Set $B=Q_{{\kappa},0}$ and compute $B = Q_{i,0}+ G_i B, i=(i_f+1,{\kappa}-1)$, and store.

\noindent \textbf{Step 3.} Compute the steady-state vector $z_s(i)$.

\noindent \textbf{Step 3.1.} Set $i=i_f+1.$ If $i=0$, find solution of $z_s(0)B=0; z_s(0)e=1.$ Start $i=1$ and compute  $z_s(i)= -z_s(i-1)\mathscr{Q}_{i-1,i}(\mathscr{Q}_{i,i} +\mathscr{Q}_{i,i+1}  G_{i})^{-1}$.

\noindent \textbf{Step 3.2.} Compute $||z_s(i)||.$ There can be three possible cases as follows.

\textit{Case 1:} If $||z_s(i)||< \epsilon_f$, set $i^*=i$ and go to Step 4.

\textit{Case 2:} If $||z_s(i)||> \epsilon_f$, and $i<{\kappa}$, increase $i$ by one and go to Step 3.1.

\textit{Case 3:} If $||z_s(i)||> \epsilon_f$, and $i={\kappa}$, set $i_0={\kappa}+s$, ${\kappa}=i_0-1+2s$ and $i_f={\kappa}$ and go to  Step 2.2.

\noindent \textbf{Step 4.} Calculate vectors $z_s({\kappa}),~ k=\overline{1,i^*}$ as $z_s({\kappa}+1) = cz_s({\kappa})$, where $c = \frac{1}{z_s(0)+z_s(1)+z_s(2)+\ldots+z_s(i^*)}$ is a normalizing constant.

\subsubsection{Modified Algorithm}
\textbf{Step 1.} Determine the initial value $i_0$. Process:
Convert the original model into Poisson-exponentially distributed model, i.e., arrival of all types of calls and arrival of catastrophe are defined by Poisson process and service, retrial and repair  processes follow exponential distribution. Further, consider a pre-defined small positive value, say $\delta,$ such that $z_s(i_0)< \delta$. Here, $z_s$ is the stationary distribution of the system and $z_s(i_0)$ is $i_0^{th}$ component of the stationary distribution.  The termination criteria  $z_s(i_0)< \delta$ represents that, for $i \geq i_0,$ the steady-state vector behaves invariantly 
 for  the particular $\delta$.

\noindent \textbf{Step 2.} Set $s$ as a multiple of $i_0,$ i.e., $s= mi_0.$ 
% If the Markov chain is ergodic, for any value of $i_0$ and for any fixed small arbitrary value $\epsilon$, there exists a finite number $s,$ the probability that the Markov chain will transit from $i$ to  $i+s$ without visiting state $i-1$ is less than $\epsilon.$

\noindent \textbf{Step 3.} Set $\epsilon_g$ and $\epsilon_f$ as accuracy levels of matrices $G_i$ and steady-state vector $z_s(i),$ respectively. 

\noindent \textbf{Step 4.} Compute matrix $G$.

\noindent \textbf{Step 4.1.} Set $G_{\kappa}^{(1)}=O$ and $G_{\kappa}^{(2)}=I$, $\kappa=i_0-1+s,$ $i_f=-1.$ 

\noindent \textbf{Step 4.2.} Compute the matrices $G_{\kappa}^{(1)}$ and $G_{\kappa}^{(2)}$ given as
\begin{center}

$G_{\kappa}^{(1)}:= -(\mathscr{Q}_{{\kappa}+1,{\kappa}+1} +\mathscr{Q}_{{\kappa}+1,{\kappa}+2}  G_{{\kappa}+1}^{(1)})^{-1}\mathscr{Q}_{{\kappa}+1,{\kappa}}$,  \\
$G_{\kappa}^{(2)}:= -(\mathscr{Q}_{{\kappa}+1,{\kappa}+1} +\mathscr{Q}_{{\kappa}+1,{\kappa}+2}  G_{{\kappa}+1}^{(2)})^{-1}\mathscr{Q}_{{\kappa}+1,{\kappa}}$.
\end{center}

\noindent \textbf{Step 4.3.} Calculate $||G_{\kappa}^{(1)}-G_{\kappa}^{(2)}||.$ There can be three possible cases. 

\textit{Case 1:} If $||G_{\kappa}^{(1)}-G_{\kappa}^{(2)}|| < \epsilon_g$, go to Step 5. 

\textit{Case 2:} If $||G_{\kappa}^{(1)}-G_{\kappa}^{(2)}|| > \epsilon_g$ and ${\kappa} \geq i_0$, set ${\kappa}={\kappa}-1$ and go to Step 4.2.

\textit{Case 3:} If $||G_{\kappa}^{(1)}-G_{\kappa}^{(2)}|| > \epsilon_g$ and ${\kappa} = i_0-1$, set $s= 2s,$ ${\kappa} =i_0-1+2s,$ $i_0={\kappa},$ and go to Step 4.1.

\noindent \textbf{Step 5.} Set $G_{\kappa} = G_{\kappa}^{(1)}.$ Compute matrices $G_i, i=(i_f+1,{\kappa}-1)$. Set $B=Q_{{\kappa},0}$ and compute $B = Q_{i,0}+ G_i B, i=(i_f+1,{\kappa}-1)$, and store.

\noindent \textbf{Step 6.} Compute the steady-state vectors $z_s(i)$.

\noindent \textbf{Step 6.1.} Set $i=i_f+1.$ If $i=0$, find solution of $z_s(0)B=0; z_s(0)e=1.$ Start $i=1$ and compute  $z_s(i)= -z_s(i-1)\mathscr{Q}_{i-1,i}(\mathscr{Q}_{i,i} +\mathscr{Q}_{i,i+1}  G_{i})^{-1}$.

\noindent \textbf{Step 6.2.} Compute $||z_s(i)||.$ There can be three possible cases

\textit{Case 1:} If $||z_s(i)||< \epsilon_f$, set $i^*=i$ and go to Step 7.

\textit{Case 2:} If $||z_s(i)||> \epsilon_f$, and $i<{\kappa}$, increase $i$ by one and go to Step 6.1.

\textit{Case 3:} If $||z_s(i)||> \epsilon_f$, and $i={\kappa}$, set $i_0={\kappa}+1$, ${\kappa}=i_0-1+2s$ and $i_f={\kappa}$ and go to  Step 4.2.

\noindent \textbf{Step 7.} Calculate vector $z_s({\kappa}), \kappa=\overline{1,i^*}$ as $z_s({\kappa}+1) = cz_s({\kappa})$, where $c = \frac{1}{z_s(0)+z_s(1)+z_s(2)+\ldots+z_s(i^*)}$ is a normalizing constant.\\

 \textbf{Advantages:} 
A very important aspect of  this modified algorithm is that the selection of initial point $i_0$ will not be random.  It is very obvious that this reduces the matrix calculations and computation time. This modified algorithm not only takes care of all the advantages of the old algorithm, but also it performs  a sequential approach to determine the search interval. Therefore,  this modified version takes care of each point to determine the search interval and no point has been discarded from the search interval while moving from one iteration to another.

\section{Performance Measures} \label{section4}
The following relevant  performance measures for  the proposed system  are calculated, after computing the  steady-state distribution $z_s$.
\begin{enumerate}

		\item The probability that there are $\mathpzc{l}$ number of retrial calls:
	\[P_{orbit}(\mathpzc{l}) = \sum_{\kappa_1=0}^{S}\sum_{\kappa_2=0}^{S-\kappa_1} {z_s} (\mathpzc{l},\kappa_1,\kappa_2,0,0)e. \]
	
	\item Expected number of retrial calls:
	\[E_{orbit} = \sum_{\mathpzc{l}=0}^{\infty} \mathpzc{l} P_{orbit}(\mathpzc{l})e. \]
	
\item The probability that $\kappa_1$ number of handoff calls are receiving service:
	\[ P_{\mathcal{E}}(\kappa_1) =   \sum_{\mathfrak{j}=1}^{S-1} {z_s} (0,0,0,\mathfrak{j},0)e + \sum_{\kappa_2=0}^{S-\kappa_1}  {z_s} (0,\kappa_1,\kappa_2,0,0)e +
\sum_{\kappa_2=0;\kappa_1+\kappa_2 \leq K}^{S-\kappa_1}  {z_s} (0,\kappa_1,\kappa_2,S,0)e \] \[+
\sum_{\kappa_2=0}^{S-\kappa_1} \sum_{i=1}^{K-\kappa_1-\kappa_2} {z_s} (0,\kappa_1,\kappa_2,S,i)e + 	\sum_{\mathpzc{l}=1}^{\infty}\sum_{\kappa_2=0}^{S-\kappa_1}  {z_s} (\mathpzc{l},\kappa_1,\kappa_2,0,0)e.\]		
		
		\item The probability that $\kappa_2$ number of new calls are receiving service:
	\[ P_{\mathcal{N}}(\kappa_2) =   \sum_{\mathfrak{j}=1}^{S-1} {z_s} (0,0,0,\mathfrak{j},0)e + \sum_{\kappa_1=0}^{S}  {z_s} (0,\kappa_1,\kappa_2,0,0)e +
\sum_{\kappa_1=0;\kappa_1+\kappa_2 \leq K}^{S}  {z_s} (0,\kappa_1,\kappa_2,S,0)e \] \[+
\sum_{\kappa_1=0}^{S} \sum_{i=1}^{K-\kappa_1-\kappa_2} {z_s} (0,\kappa_1,\kappa_2,S,i)e + 	\sum_{\mathpzc{l}=1}^{\infty}\sum_{\kappa_1=0}^{S}  {z_s} (\mathpzc{l},\kappa_1,\kappa_2,0,0)e.\]

	\item The probability that $i$ number of emergency calls are receiving service:
	\[ P_{\mathcal{E}}(i) =   \sum_{\mathfrak{j}=1}^{S} {z_s} (0,0,0,\mathfrak{j},0)e + \sum_{\kappa_1=0}^{S} \sum_{\kappa_2=0}^{S-\kappa_1} {z_s} (0,\kappa_1,\kappa_2,0,0)e +
\sum_{\kappa_1=0}^{S} \sum_{\kappa_2=0;\kappa_1+\kappa_2 \leq K}^{S-\kappa_1} {z_s} (0,\kappa_1,\kappa_2,S,0)e \] \[+
\sum_{\kappa_1=0}^{S} \sum_{\kappa_2=0;\kappa_1+\kappa_2 + i \leq K}^{S-\kappa_1} {z_s} (0,\kappa_1,\kappa_2,S,i)e + 	\sum_{\mathpzc{l}=1}^{\infty}\sum_{\kappa_1=0}^{S}\sum_{\kappa_2=0}^{S-\kappa_1}  {z_s} (\mathpzc{l},\kappa_1,\kappa_2,0,0)e.\]

	\item The probability that the system is under repair:
	\[ P_R =   \sum_{\mathfrak{j}=1}^{S-1} {z_s} (0,0,0,\mathfrak{j},0)e + \sum_{\kappa_1=0}^{K} \sum_{\kappa_2=0}^{K-\kappa_1} \sum_{i=0}^{K-\kappa_1-\kappa_2} {z_s} (0,\kappa_1,\kappa_2,S,i)e. \]
	
		\item The dropping probability of a handoff call:
\begin{itemize}
    \item[-] in normal scenario:
    	\[ P_d^n =  \frac{1}{\lambda_{\mathcal{H}}} \Big( \sum_{\mathpzc{l}=0}^{\infty} {z_s} (\mathpzc{l},S,0,0,0) (C_{\mathcal{H}}\otimes I_{
	{\scriptstyle{L_2T^{M_{\mathcal{H}}}_{\kappa_1}T^{N}_{\mathpzc{l}}}}}) e\Big) \]
	\item[-] in catastrophic scenario:
	\[ P_d^c =  \frac{1}{\lambda_{\mathcal{N}}} \Big( \sum_{\kappa_1=0}^{K} \sum_{\kappa_2=0}^{K-\kappa_1}{z_s} (0,\kappa_1,\kappa_2,S,K-\kappa_1-\kappa_2) (C_{\mathcal{H}}\otimes I_{
	{\scriptstyle{L_2T^{M_{\mathcal{H}}}_{\kappa_1}T^{M_{\mathcal{N}}}_{\kappa_2}T^{M_{\mathcal{E}}}_{K-\kappa_1-\kappa_2}T^{R}_{S}}}}) e\Big). \]
	
\end{itemize}

	\item The blocking probability of a new call in catastrophic scenario:
	\[ P_b^c =  \frac{1}{\lambda_{\mathcal{N}}} \Big( \sum_{\kappa_1=0}^{K} \sum_{\kappa_2=0}^{K-\kappa_1}{z_s} (0,\kappa_1,\kappa_2,S,K-\kappa_1-\kappa_2) (C_{\mathcal{N}}\otimes I_{
	{\scriptstyle{L_2T^{M_{\mathcal{H}}}_{\kappa_1}T^{M_{\mathcal{N}}}_{\kappa_2}T^{M_{\mathcal{E}}}_{K-\kappa_1-\kappa_2}T^{R}_{S}}}}) e\Big). \]

	\item The blocking probability of an emergency call:
	\[ P_e =  \frac{1}{\lambda_{\mathcal{E}}} \Big( {z_s} (0,0,0,S,K) (C_{\mathcal{E}}\otimes I_{
	{\scriptstyle{L_2T^{M_{\mathcal{E}}}_{K}T^{R}_{S}}}}) e\Big). \]

    \item Rate of losses due to the occurrence of catastrophe:
	\[\alpha_{f} = \alpha \sum_{\mathpzc{l}=0}^{\infty} \sum_{\kappa_1=0}^{S} \sum_{\kappa_2=0; \kappa_1=\kappa_1 \neq 0}^{S-\kappa_1}  {z_s} (\mathpzc{l},\kappa_1,\kappa_2,0,0)  (D_1\otimes I_{
	{\scriptstyle{L_1T^{M_{\mathcal{H}}}_{\kappa_1}T^{M_{\mathcal{N}}}_{\kappa_2}T^{N}_{\mathpzc{l}}}}})  e. \]     
	
	\item The probability that an arriving handoff call preempts the service of an ongoing new call in normal scenario:
	\[P_{preempt}^{new} = \frac{1}{\lambda_{\mathcal{H}}} \sum_{\mathpzc{l}=0}^{\infty}\sum_{\kappa_1=0}^{K_2-1}  {z_s} (\mathpzc{l},\kappa_1,S-\kappa_1,0,0)  (C_{\mathcal{H}}\otimes I_{
	{\scriptstyle{L_2T^{M_{\mathcal{H}}}_{\kappa_1}T^{M_{\mathcal{N}}}_{S-\kappa_1}T^{N}_{\mathpzc{l}}}}})  e. \]  
	
	\item The probability that an arriving emergency call preempts the service of an ongoing handoff/new call in catastrophic scenario:
	\[P_{preempt}^{emr} = \frac{1}{\lambda_{\mathcal{E}}} \Big(  {z_s} (0,0,K,S,0) + \sum_{\kappa_1=1}^{K} \sum_{\kappa_2=0\& K-\kappa_1-\kappa_2<K_1}^{K-\kappa_1}  {z_s} (0,\kappa_1,\kappa_2,S,K-\kappa_1-\kappa_2) \Big)  \Big(C_{\mathcal{E}}\otimes I_{
	{\scriptstyle{L_2T^{M_{\mathcal{H}}}_{\kappa_1}T^{M_{\mathcal{N}}}_{\kappa_2}T^{M_{\mathcal{E}}}_{K-\kappa_1-\kappa_2}T^{R}_{S}}}} \Big) e. \]   
	
	\item The intensity by which a retrial call is successfully connected to an available channel:
	\[ \theta_r^{succ} =  \sum_{\mathpzc{l}=1}^{\infty}\sum_{\kappa_1=0}^{S}\sum_{\kappa_2=0}^{S-\kappa_1} \theta {z_s} (\mathpzc{l},\kappa_1,\kappa_2,0,0) ( I_{
	{\scriptstyle{L_1L_2T^{M_{\mathcal{H}}}_{\kappa_1}T^{M_{\mathcal{N}}}_{\kappa_2}T^{N}_{\mathpzc{l}-1}}}} \otimes( \Gamma^{0}(2)\otimes \beta_{\mathcal{N}}))e.\]

\end{enumerate}
\section{Numerical Illustration} \label{section5}
In this section, the qualitative behaviour of the proposed model is explored through a few experiments. All the numerical experiments have been conducted by considering $\delta=10^{-12}$, $\epsilon_g=10^{-10} $ and $\epsilon_f=10^{-10}$ to compute the steady-state distribution.
 The matrices for the $\textrm{\it M\!M\!A\!P}$ and $\textrm{\it P\!H}$ distribution parameters are referred from \cite{dudin2016analysis} as follows.
 
 \noindent (a) \textbf{Normal Scenario:}
{\small{
\begin{align}
	\nonumber
	C_{0}= \begin{pmatrix}
		-0.8109843 & 0\\
		0 & -0.02632213
	\end{pmatrix},C_{\mathcal{H}}= \begin{pmatrix}
		0.201398 & 0.0013479\\
		0.003665 & 0.0029153
	\end{pmatrix}, C_{\mathcal{N}}= \begin{pmatrix}
	0.6041 & 0.0040439\\
		0.01099& 0.008745
	\end{pmatrix}.
\end{align}}}
\noindent (b) \textbf{Catastrophic Scenario:}
{\small{
\begin{align}
	\nonumber
	C_{0}= \begin{pmatrix}
		-0.810 & 0\\
		0 & -0.026
	\end{pmatrix},C_{\mathcal{H}}= \begin{pmatrix}
		0.20 & 0.0013\\
		0.003 & 0.002
	\end{pmatrix}, C_{\mathcal{N}}= \begin{pmatrix}
	0.30 & 0.0020\\
		0.005 & 0.004
	\end{pmatrix},C_{\mathcal{E}}= \begin{pmatrix}
		0.30 & 0.0020\\
		0.005 & 0.004
	\end{pmatrix}.
\end{align}}}

 The  correlation coefficients for both types of calls are $C_{r}^{(1)}=C_{r}^{(2)}=0.2$ and the  variation coefficients  for both types of calls are $C_{r}^{(1)}=C_{r}^{(2)}=12.34.$ 
The arrival rates are $\lambda_{\mathcal{H}}= 0.15, \lambda_{\mathcal{N}} =0.45$ in normal scenario and $\lambda_{\mathcal{H}}= 0.15, \lambda_{\mathcal{N}} =0.23$ and $\lambda_{\mathcal{E}} =0.22$ in catastrophic scenario.
Let $\textrm{\it P\!H}$ distribution parameters for the service rates of a handoff, a new call, and an emergency call be
\begin{align}
	\nonumber
	\beta_{\mathcal{H}}= \begin{pmatrix}
	0.05, &0.95
	\end{pmatrix},~~ A_{\mathcal{H}}= \begin{pmatrix}
		-0.031 & 0\\
		0 & -2.4
	\end{pmatrix},~\beta_{\mathcal{N}}= \begin{pmatrix}
		0.1,& 0.9
	\end{pmatrix}, ~~ A_{\mathcal{N}}= \begin{pmatrix}
		-0.033 & 0\\
		0 & -2.52
	\end{pmatrix},
\end{align}

\begin{align}
	\nonumber
\textrm{and}~~~	\beta_{\mathcal{E}}= \begin{pmatrix}
	0, &1
	\end{pmatrix},~~ A_{\mathcal{E}}= \begin{pmatrix}
		-1 & 0\\
		0 & -1
	\end{pmatrix}.
\end{align}
The fundamental service rates are  $\mu_{\mathcal{H}}=0.5,\mu_{\mathcal{N}}=0.3$ and $ \mu_{\mathcal{E}}=0.5$. The retrial rate of a retrial call, following  $\textrm{\it PH}$ distribution, is given by the parameters (refer, \cite{artalejo2007modelling})
\begin{align*}
	\nonumber
	\gamma= \begin{pmatrix}
		1, & 0
	\end{pmatrix},~~~ \Gamma = \begin{pmatrix}
		-2 & 2\\
		0 & -2
	\end{pmatrix}~~\textrm{and}~~\theta=1.
\end{align*}

 To demonstrate the feasibility of the developed model,
some interesting observations of the proposed system are described through the following numerical experiments. These experiments will present the behaviour of  performance measures with respect to  arrival, service and retrial rates.\\

 \noindent \textbf{Experiment 1:} The objective here is to analyze the impact of arrival rate ($\lambda_{\mathcal{H}}$) and service rate ($\mu_{\mathcal{H}}$) of handoff call  over the  dropping probability of handoff call in normal scenario ($P_d^n$).
 
 It can be observed from the Figures \ref{fig:1a} and \ref{fig:1b} that $P_d^n$ exhibits increasing behaviour with respect to $\lambda_{\mathcal{H}}$. In Figure \ref{fig:1a}, the value of $P_d^n$ decreases as the value of threshold number for preemption of new calls ($K_2$) increases. Similarly, a decrement can be seen in the value of $P_d^n$ with the increasing number of total number of channels ($S$) in the system. The explanation for this particular behaviour can be given as follows. When  handoff calls arrive frequently in the system, all the channels may be occupied by handoff calls and consequently the arriving  handoff calls may be dropped. If the threshold level $K_2$ is increased, more handoff calls will be able to receive the service by preempting the service of ongoing new calls. Therefore, $P_d^n$ becomes an increasing function of $\lambda_{\mathcal{H}}$ and a decreasing function with respect to $K_2$. Similarly, if the value of total number of channels $S$ in the system increases, more handoff calls will be served and consequently, $P_d^n$ decreases. 
 
 Figures \ref{fig:2a} and \ref{fig:2b} exhibit an opposite decreasing behaviour of $P_d^n$ with respect to $\mu_{\mathcal{H}}$. As the service of handoff calls increases, the probability of handoff calls getting service also increases, hence $P_d^n$ decreases. As $S$ and $K_2$ increases along with $\mu_{\mathcal{H}}$, it can be seen that $P_d^n$ attains very small value. \\
 
 \noindent \textbf{Experiment 2:} The main purpose of this experiment is to  exhibit the impact of  $\lambda_{\mathcal{H}}$, $\mu_{\mathcal{H}}$, $K_2$ and $S$ over preemption probability for new calls in normal scenario over $P_{preempt}^{new}$.

It can be seen from the Figures \ref{fig:3a} and \ref{fig:3b} that the values of $P_{preempt}^{new}$  for different values of $K_2$ and $S$, first increase, and then decrease. The cause for this behavior of $P_{preempt}^{new}$  lies in the following explanation.
When $\lambda_{\mathcal{H}}$  is relatively small, an arriving handoff call often finds at least one channel available, and consequently the ongoing service of a new call is not preempted by the arriving handoff call. As $\lambda_{\mathcal{H}}$ increases, the number of handoff calls also increase in the system. If an  arriving handoff call finds all the channels occupied and at least one of them is serving a new call, the service of that new call will be preempted by the arriving handoff call. Hence, $P_{preempt}^{new}$  increases and reaches maximum at some value of $\lambda_{\mathcal{H}}$. Further, the decreasing behaviour of $P_{preempt}^{new}$   is explained by the fact that, with the increment in $\lambda_{\mathcal{H}}$,  all the channels are occupied with handoff calls. Thus, the number of new calls in the service decreases and the probability that an arriving handoff call preempts the service of a new call decreases. If the value of $K_2$ increases, it will increase the preemption of ongoing new calls, consequently $P_{preempt}^{new}$ increases whereas if the value of $S$ is increased for a fix value of $\lambda_{\mathcal{H}}$, $P_{preempt}^{new}$  decreases.

Figures \ref{fig:4a} and \ref{fig:4b} shows that $P_{preempt}^{new}$ is a decreasing function of $\mu_{\mathcal{H}}$ for fixed values of $K_2$ and $S$. If the handoff calls are served at  an increasing rate, the probability that handoff calls preempts the service of ongoing new calls decreases, hence $P_{preempt}^{new}$ decreases. If the value of $K_2$ increases more ongoing new calls will be preempted, therefore $P_{preempt}^{new}$ increases whereas if $S$ increases, $P_{preempt}^{new}$ decreases.\\

  \noindent \textbf{Experiment 3:} In this experiment, the behaviour of blocking probability for emergency call $P_e$ is illustrated with respect to arrival rate of emergency call ($\lambda_{\mathcal{E}}$) and service rate of emergency call ($\mu_{\mathcal{E}}$) for different value of backup channels ($K$) and threshold value of preemption for emergency call ($K_1$).
 
 It can be observed from Figures \ref{fig:5a} and \ref{fig:5b} that $P_e$ is an increasing function of  $\lambda_{\mathcal{E}}$ for fixed $K$ and $K_1$. An intuitive explanation for this finding can easily be given as follows. The increment in the value of $\lambda_{\mathcal{E}}$ leads to  an increment in the blocking probability of emergency calls as all of the channels will be occupied with the emergency calls only, hence $P_e$ decreases.  If $K$ is increased in the system, the probability  for emergency calls to obtain the service also increase and consequently $P_{e}$ decreases. Similar behaviour can be observed when the value of threshold $K_1$ increases. When $K_1$ increases, more emergency calls will be able to preempt the service of ongoing new/handoff calls and start service in its place. Thus, $P_e$ decreases with respect to $\lambda_{\mathcal{E}}$ for fixed values of $K_1$.

 On the contrary, the decreasing behaviour of $P_e$ can be observed from \ref{fig:5c} and \ref{fig:5d} with respect to $\mu_{\mathcal{E}}$. When $\mu_{\mathcal{E}}$ increases in the system, the emergency calls are served with increasing rate, consequently the blocking of emergency calls will be reduced and $P_{e}$ will decrease. If $K$ and $K_1$ are increased in the system, the probability  for emergency calls to obtain the service also increases and consequently, $P_{e}$ decreases.\\
 
  \noindent \textbf{Experiment 4:} The objective here is to demonstrate the impact of arrival rate of emergency call ($\lambda_{\mathcal{E}}$) and service rate of emergency call ($\mu_{\mathcal{E}}$) for different value of backup channels ($K$) and threshold value of preemption for emergency call ($K_1$)  over the  blocking probability of new call in catastrophic scenario ($P_d^c$).
  
  Figures \ref{fig:6a} and \ref{fig:6b} show the behaviour of blocking probability of new call $P_b^c$ with respect to $\lambda_{\mathcal{E}}$ for different values of $K$ and $K_1$.  $P_b^c$ increases as  $\lambda_{\mathcal{E}}$ increases in the system for fixed $K$. As the number of emergency calls increases in the system, less channels will be available for handoff and new calls. As a consequence, an arriving handoff/new call will be dropped from the system and $P_b^c$ increases. If the value of $K$ is increased, the probability that new/handoff calls will receive service for fixed $\lambda_{\mathcal{E}}$ increases. Whereas, it can be observed that if the value of $K_1$ is increased, there is negligible impact over the value of $P_b^c$ in the system.

  Figures \ref{fig:6c} and \ref{fig:6d} represent that $P_b^c$ decreases as $\mu_{\mathcal{E}}$ increases.  The decreasing behaviour of $P_b^c$ can be observed as $\mu_{\mathcal{E}}$ increases in the system because the emergency calls will be served with an increasing rate, and more channels will be available for other calls, consequently $P_b^c$ decreases. If $K$  and $K_1$ increases in the system, the probability  for handoff/new calls to obtain the service also increase and consequently $P_b^c$ decreases.\\
 
  \noindent \textbf{Experiment 5:} The main purpose of this experiment is to observe the behaviour of preemption probability for emergency calls in catastrophic scenario with respect to the arrival rate of emergency call ($\lambda_{\mathcal{E}}$) and service rate of emergency call ($\mu_{\mathcal{E}}$).

 Figures \ref{fig:7a} and \ref{fig:7b} represent the behaviour of preemption probability for emergency call $P_{preempt}^{emr}$ over $\lambda_{\mathcal{E}}$. Initially, for small values of  $\lambda_{\mathcal{E}}$, $P_{preempt}^{emr}$ increases which shows that  due to the arrival of emergency calls, more handoff/new calls will be preempted from the system. After a certain value of $\lambda_{\mathcal{E}}$, $P_{preempt}^{emr}$ decreases as $\lambda_{\mathcal{E}}$ increases. Since   after a certain time period, most of the channels will be occupied with the emergency calls only, consequently there will be no preemption of handoff/new calls in the system.  If the value of $K$ is increased in the system, a decreasing behaviour of $P_{preempt}^{emr}$ can be observed. On the contrary, if $K_1$ is increased, $P_{preempt}^{emr}$ increases, as more new/handoff calls will be preempted with the arrival of emergency calls. 
 
 From the Figures \ref{fig:7c} and \ref{fig:7d}, the decreasing behaviour of $P_{preempt}^{emr}$ can be observed with respect to $\mu_{\mathcal{E}}$. An intuitive explanation for this finding can be given as follows. When the handoff calls are served at increasing rate, the chances of their preempting the new calls will be reduced, therefore, $P_{preempt}^{emr}$ will decrease. As $K$ increases in the system, $P_{preempt}^{emr}$ also decreases. The opposite behaviour can be seen when $K_1$ increases in the system. The increasing value of $K_1$ implies that the service of ongoing new/handoff call will be preempted with the arrival of emergency calls. Therefore, $P_{preempt}^{emr}$ increases when $K_1$ increases in the system.
 
  All these observations are the main motivation for the formulation of the  multi-objective optimization problem illustrated in Section \ref{section6}.

%######################nd ########
\begin{figure}[htp]
	\centering
		\subfigure[$P_d^n$ versus  $\lambda_{\mathcal{H}}$ for $S=6$]
	{\includegraphics[trim= 3cm 0.1cm 2.0cm 0.4cm, height = 5.55cm,width = 0.5\textwidth]{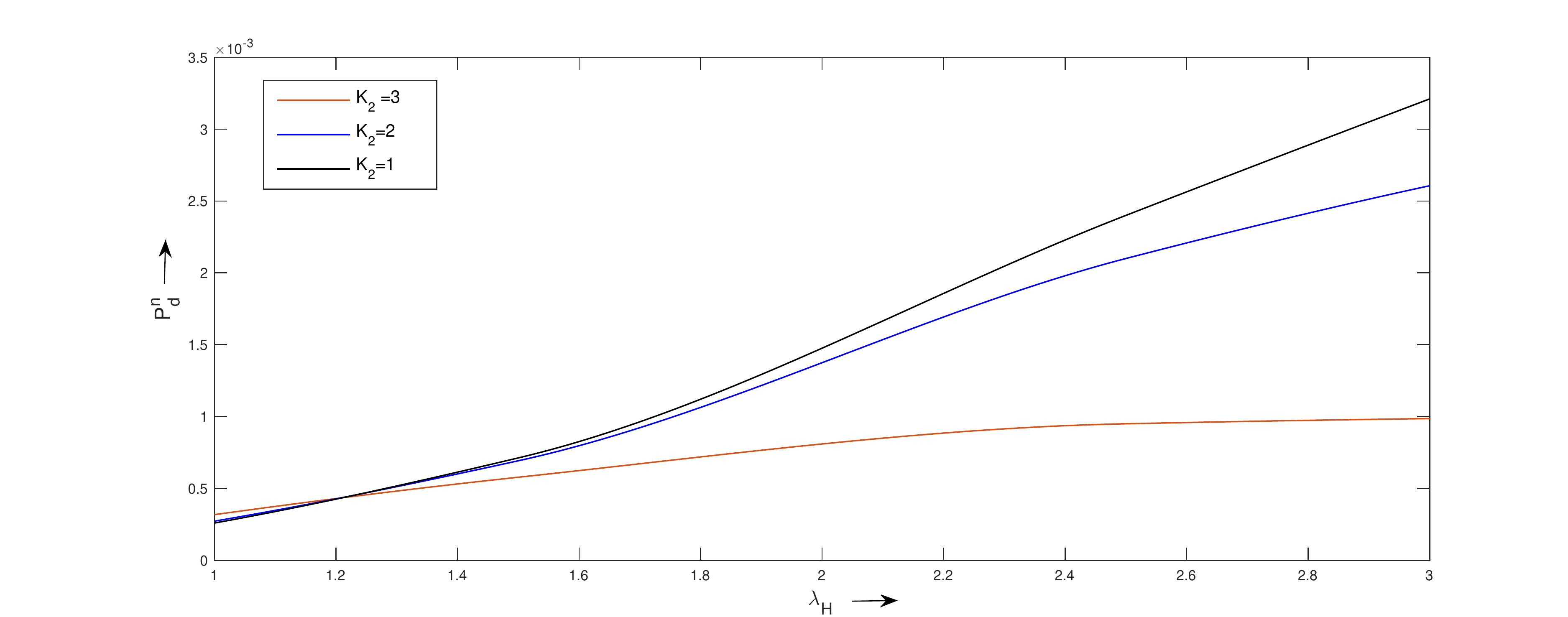}
		\label{fig:1a}}%
	\subfigure[$P_d^n$ versus  $\lambda_{\mathcal{H}}$ for $K_2=3$]
	{\includegraphics[trim= 1cm 0.1cm 3.0cm 0.4cm, height = 5.55cm,width = 0.5\textwidth]{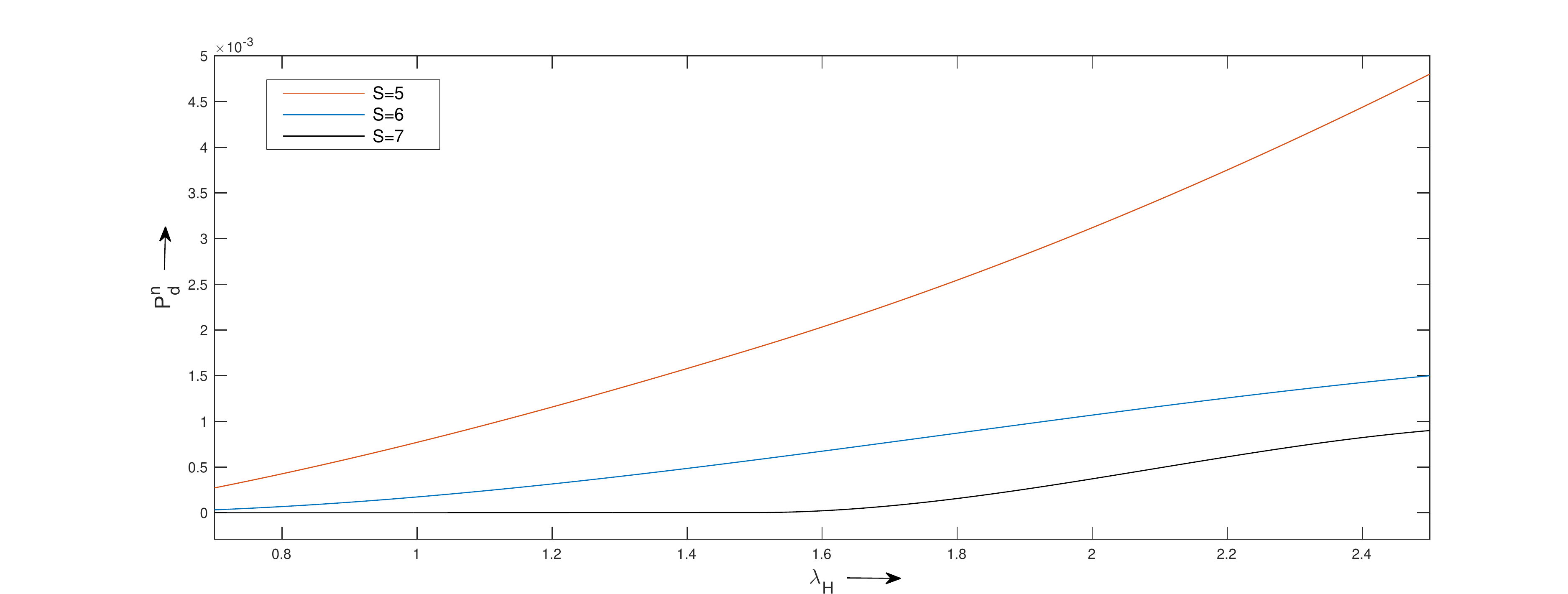}
		\label{fig:1b}}%
		
		\subfigure[$P_d^n$ versus  $\mu_{\mathcal{H}}$ for $S=6$]
	{\includegraphics[trim= 3cm 0.1cm 2.0cm 0.4cm, height = 5.55cm,width = 0.5\textwidth]{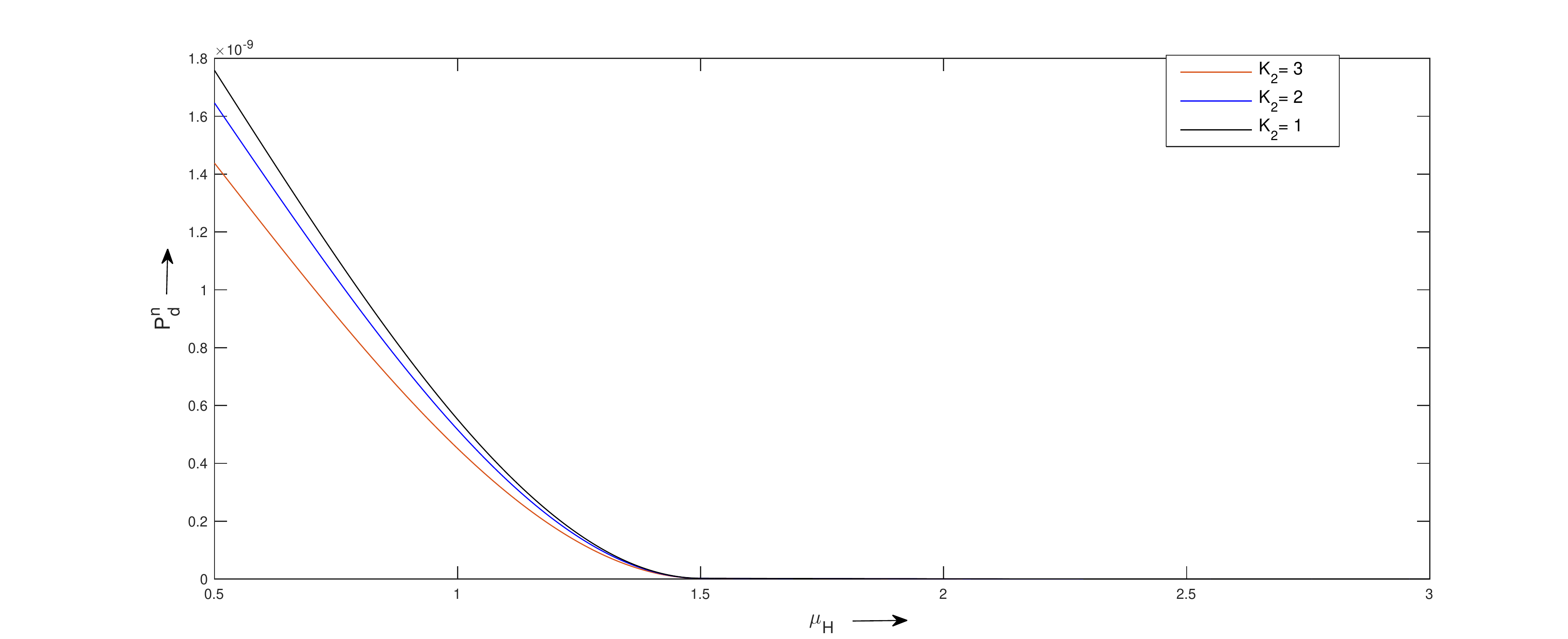}
	\label{fig:2a}}%
	\subfigure[$P_d^n$ versus  $\mu_{\mathcal{H}}$ for $K_2=3$]
	{\includegraphics[trim= 1cm 0.1cm 2.0cm 0.4cm, height = 5.55cm,width = 0.5\textwidth]{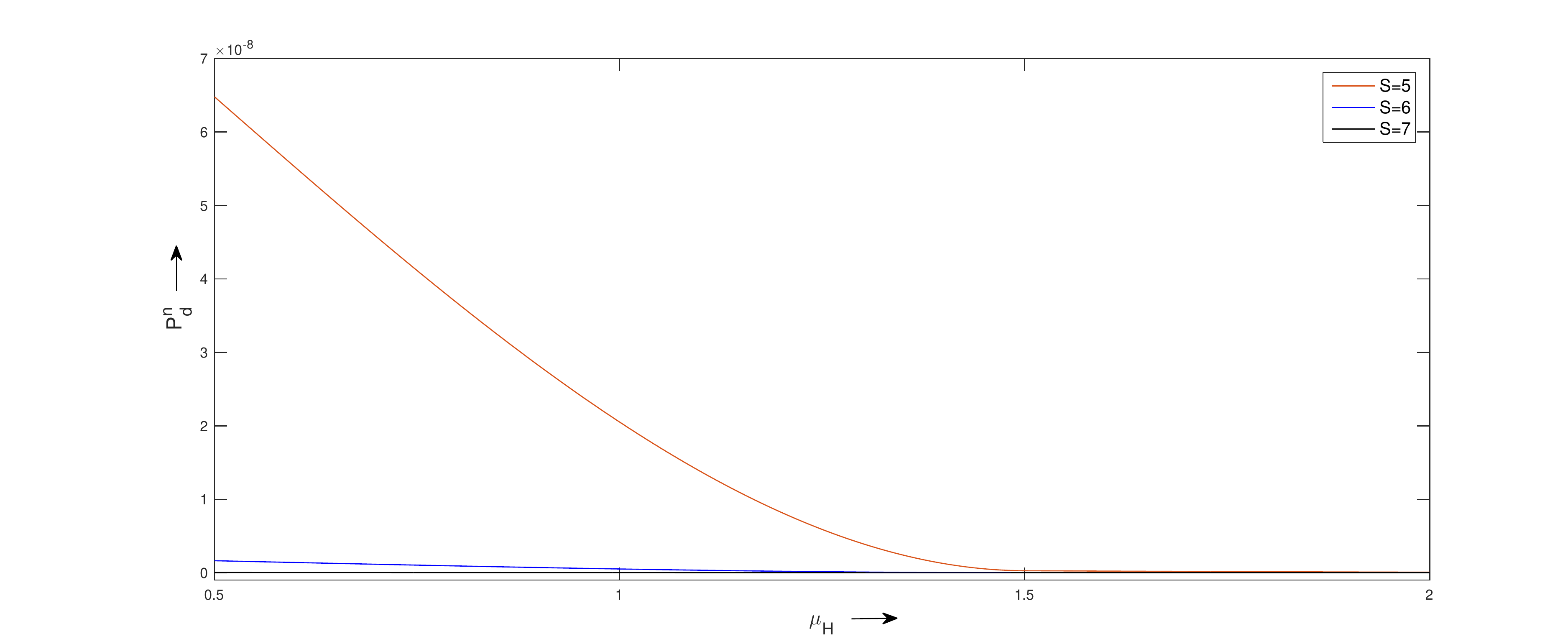}
		\label{fig:2b}}%
		
	\caption{Dependence of the dropping probability $P_d^n$  over arrival rate of a handoff call $\lambda_{\mathcal{H}}$ and service rate of a handoff call $\mu_{\mathcal{H}}$ for  $S=6$, and $K_2=3$, respectively. }
		
	\label{fig:Pb_Pd1}
\end{figure}
%======================================

\begin{figure}[htp]
	\centering
		\subfigure[$P_{preempt}^{new}$ versus  $\lambda_{\mathcal{H}}$ for $S=6$]
		{\includegraphics[trim= 2cm 0.1cm 2.5cm 0.4cm, height = 5.55cm,width = 0.5\textwidth]{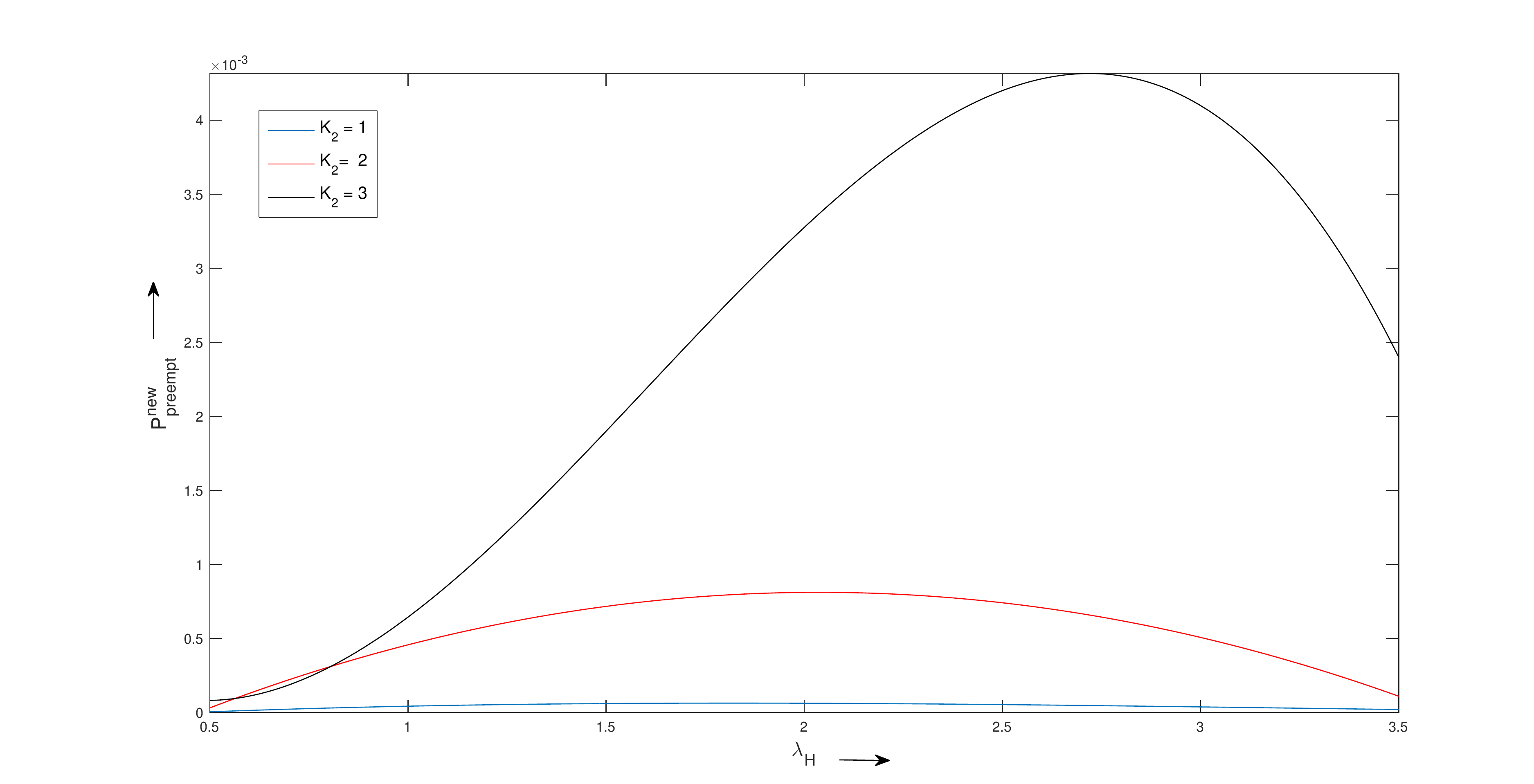}
		\label{fig:3a}}%
		\subfigure[$P_{preempt}^{new}$ versus  $\lambda_{\mathcal{H}}$ for $K_2=3$]
	{\includegraphics[trim= 1cm 0.1cm 1.8cm 0.4cm, height = 5.55cm,width = 0.5\textwidth]{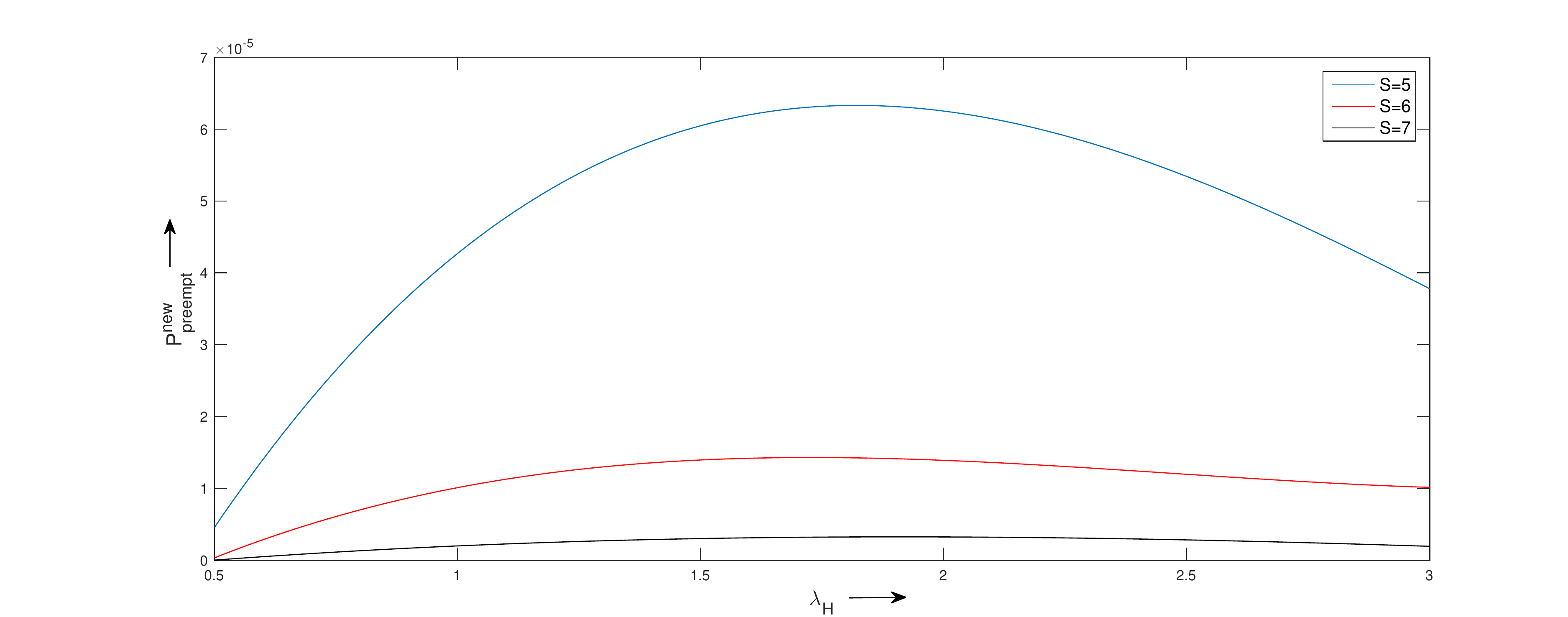}
		\label{fig:3b}}%

	\subfigure[$P_{preempt}^{new}$ versus  $\mu_{\mathcal{H}}$ for $S=6$]
		{\includegraphics[trim= 2cm 0.1cm 2.5cm 0.4cm, height = 5.55cm,width = 0.5\textwidth]{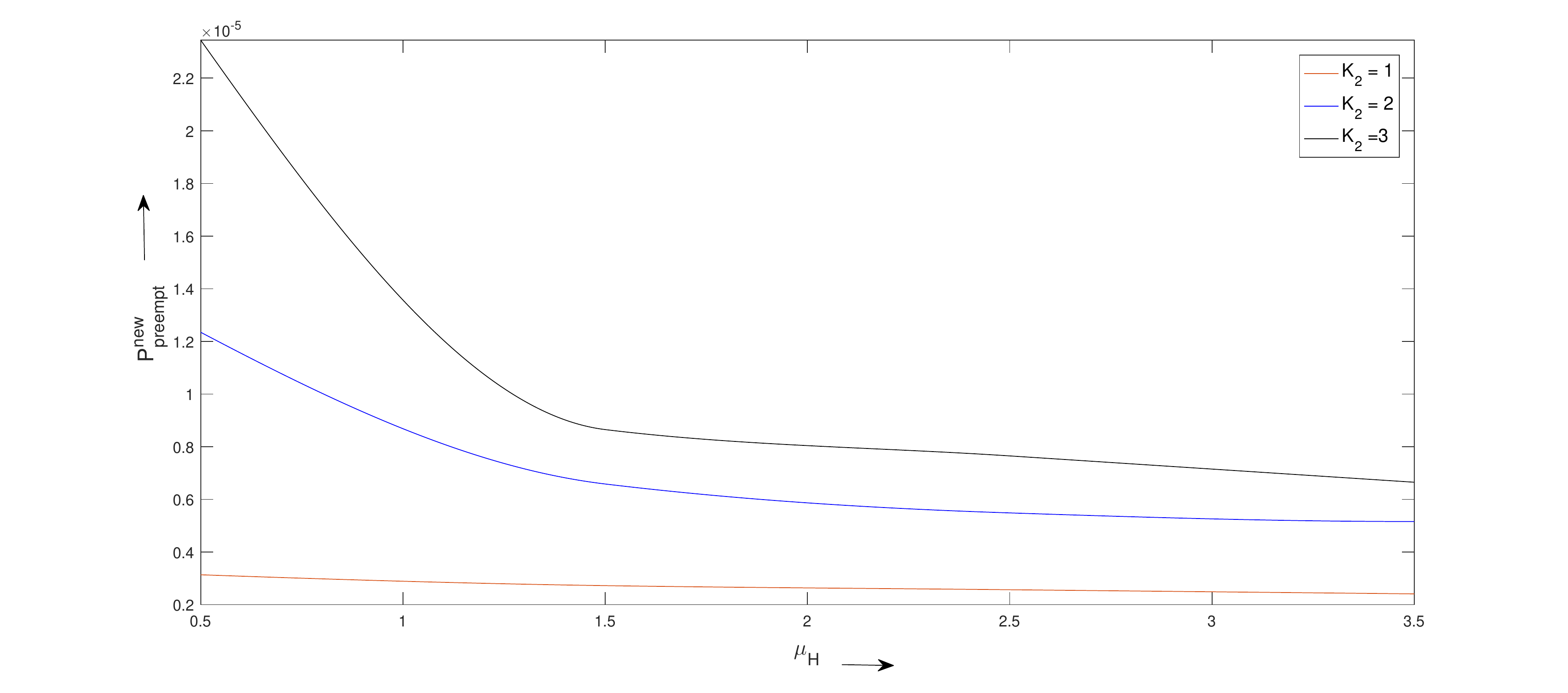}
		\label{fig:4a}}%
		\subfigure[$P_{preempt}^{new}$ versus  $\mu_{\mathcal{H}}$ for $K_2=3$]
{\includegraphics[trim= 1cm 0.1cm 1.8cm 0.4cm, height = 5.55cm,width = 0.5\textwidth]{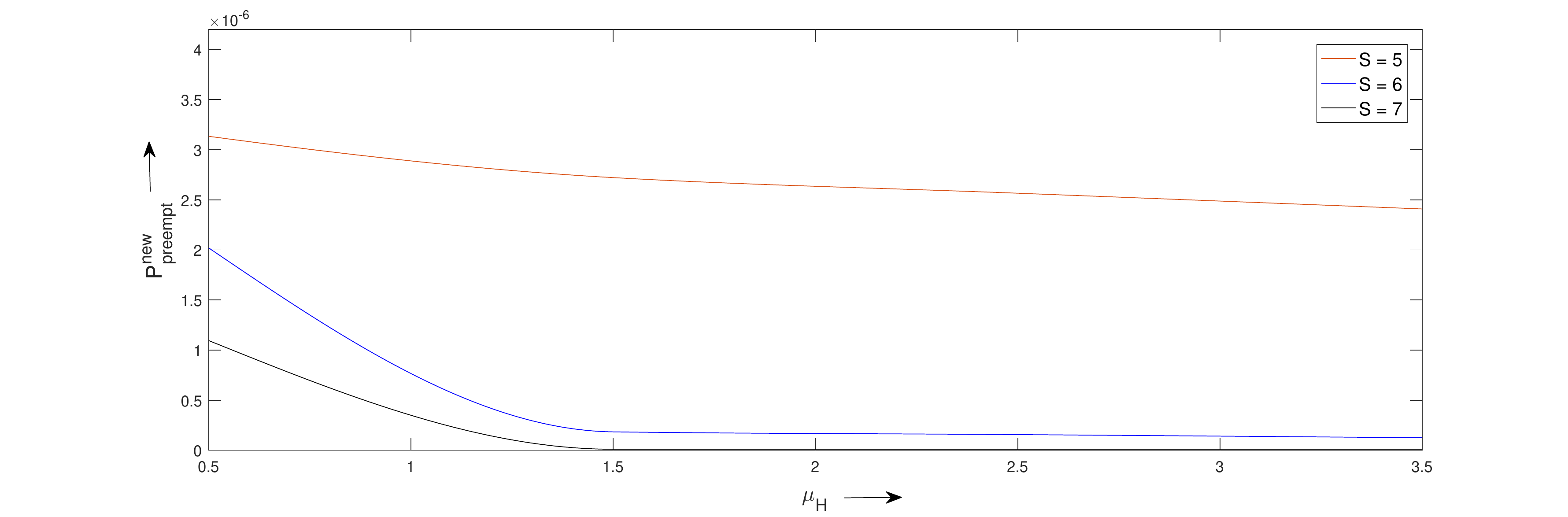}
		\label{fig:4b}}%
		
		\caption{Dependence of the preemption probability $P_{preempt}^{new}$  over arrival rate of a handoff call $\lambda_{\mathcal{H}}$ and service rate of a handoff call $\mu_{\mathcal{H}}$ for $S=6$ and $K_2=3$, respectively. }
		
	\label{fig:Pb_Pd3}
\end{figure}
%======================================

%========================================
\begin{figure}[htp]
	\centering
		\subfigure[$P_{e}$ versus  $\lambda_{\mathcal{E}}$ for $S=5, K_1=2$]
	{\includegraphics[trim= 3cm 0.1cm 1.1cm 0.4cm, height = 5.55cm,width = 0.5\textwidth]{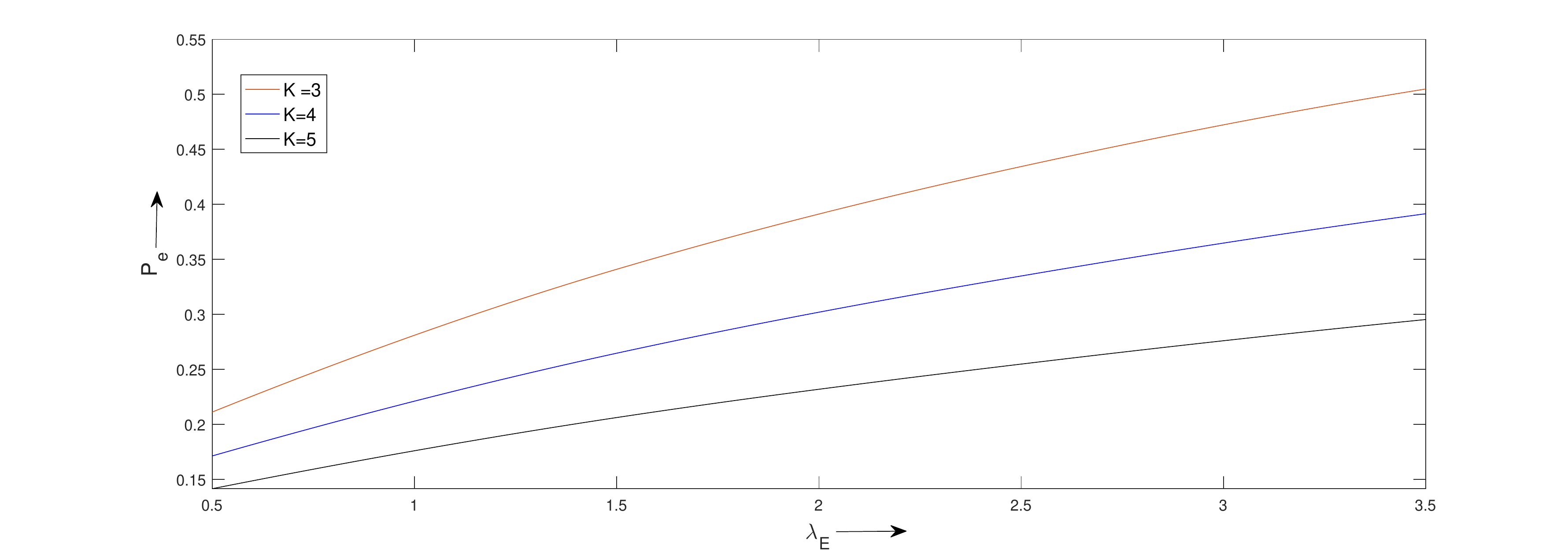}
		\label{fig:5a}}%
	\subfigure[$P_{e}$ versus  $\lambda_{\mathcal{E}}$ for $S=5,K=4$]
	{\includegraphics[trim= 2cm 0.1cm 2.5cm 0.4cm, height = 5.55cm,width = 0.5\textwidth]{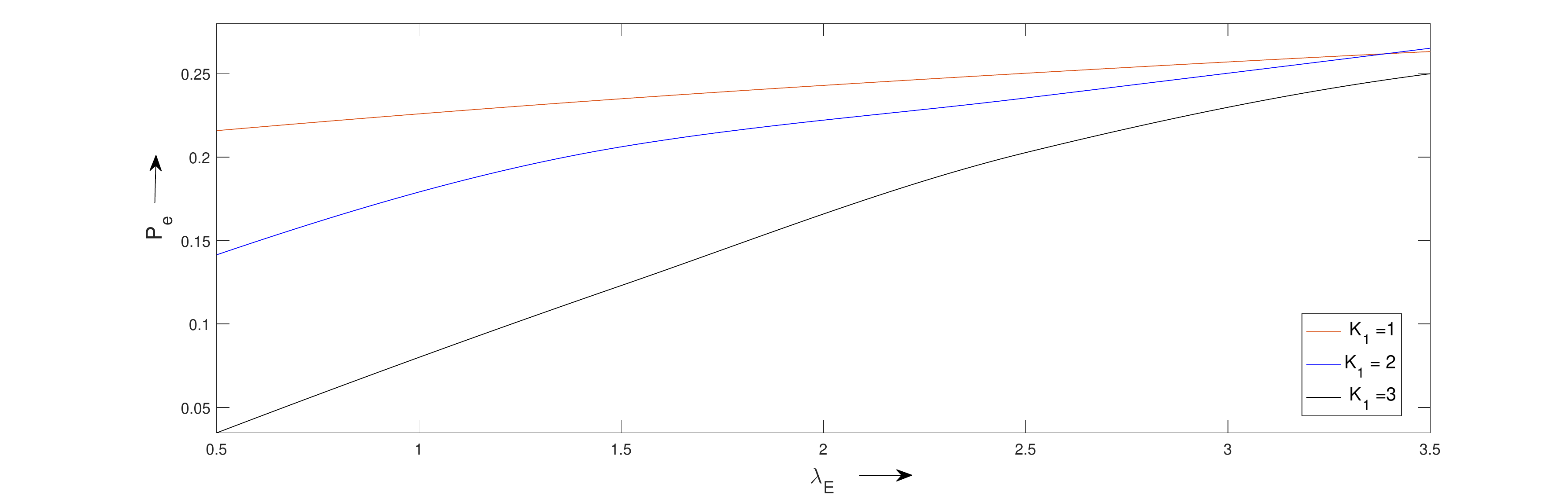}
		\label{fig:5b}}%
		
		\subfigure[$P_{e}$ versus  $\mu_{\mathcal{E}}$ for $S=5, K_1=2$]
	{\includegraphics[trim= 3cm 0.1cm 1.1cm 0.4cm, height = 5.55cm,width = 0.5\textwidth]{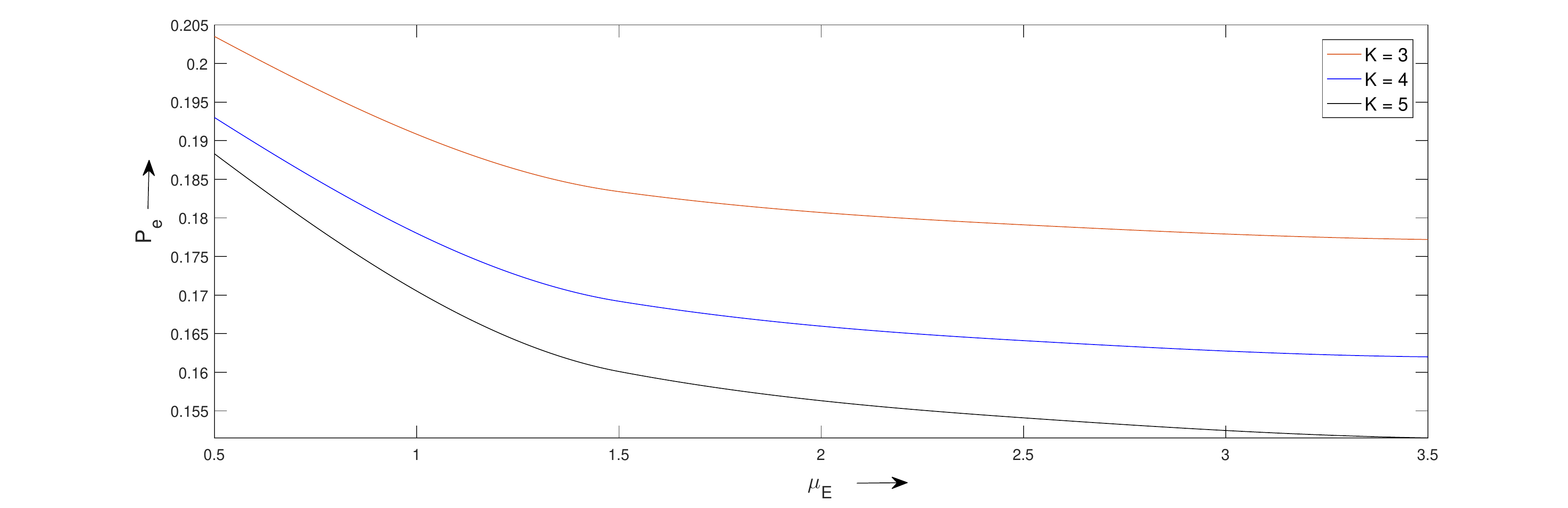}
		\label{fig:5c}}%
	\subfigure[$P_{e}$ versus  $\mu_{\mathcal{E}}$ for $S=5,K=4$]
	{\includegraphics[trim= 2cm 0.1cm 2.5cm 0.4cm, height = 5.55cm,width = 0.5\textwidth]{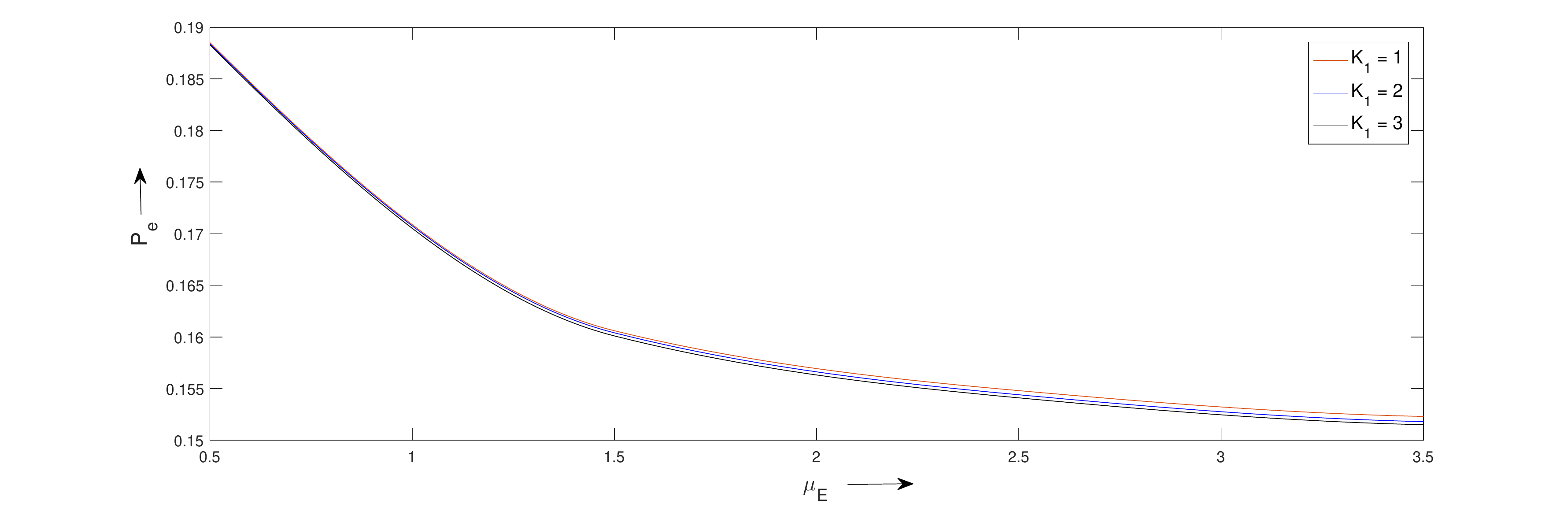}
		\label{fig:5d}}%
			\caption{Dependence of the blocking probability for emergency call $P_e$  over arrival rate of an emergency call $\lambda_{\mathcal{E}}$ and service rate of an emergency call $\mu_{\mathcal{E}}$. }
		
	\label{fig:Pb_Pd5}
\end{figure}
%======================================

%========================================
\begin{figure}[htp]
	\centering
		\subfigure[$P_{b}^c$ versus  $\lambda_{\mathcal{E}}$ for $S=5, K_1=2$]
	{\includegraphics[trim= 3cm 0.1cm 1.1cm 0.4cm, height = 5.55cm,width = 0.5\textwidth]{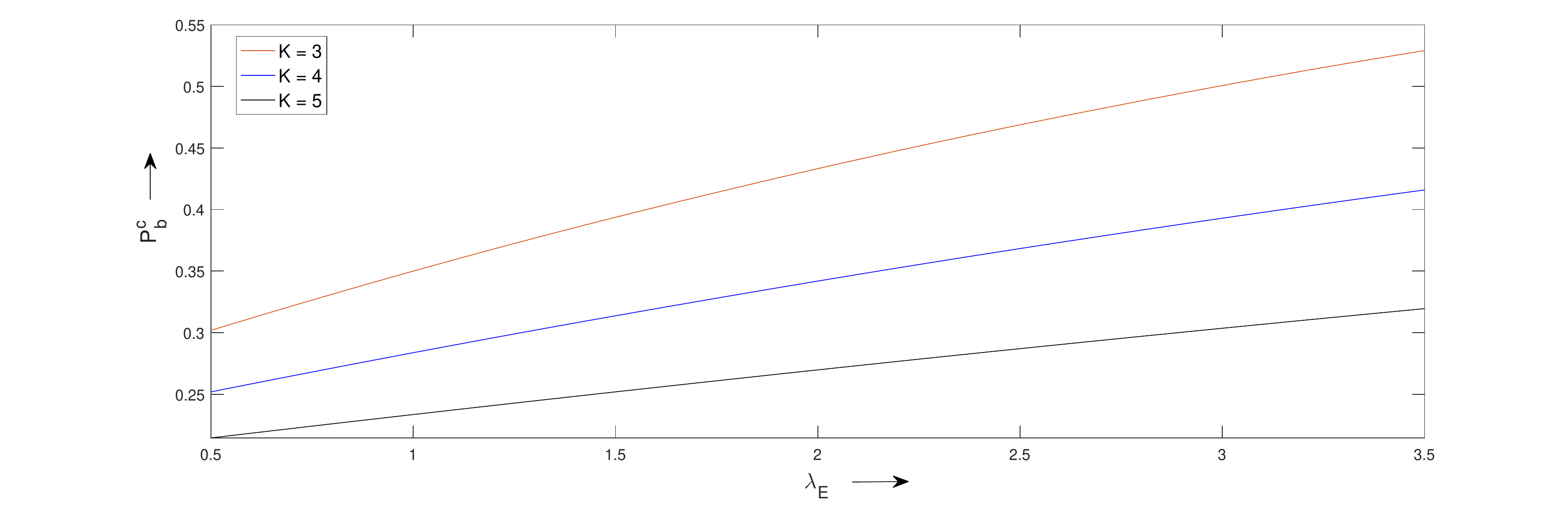}
		\label{fig:6a}}%
	\subfigure[$P_{b}^c$ versus  $\lambda_{\mathcal{E}}$ for $S=5,K=4$]
	{\includegraphics[trim= 2cm 0.1cm 2.5cm 0.4cm, height = 5.55cm,width = 0.5\textwidth]{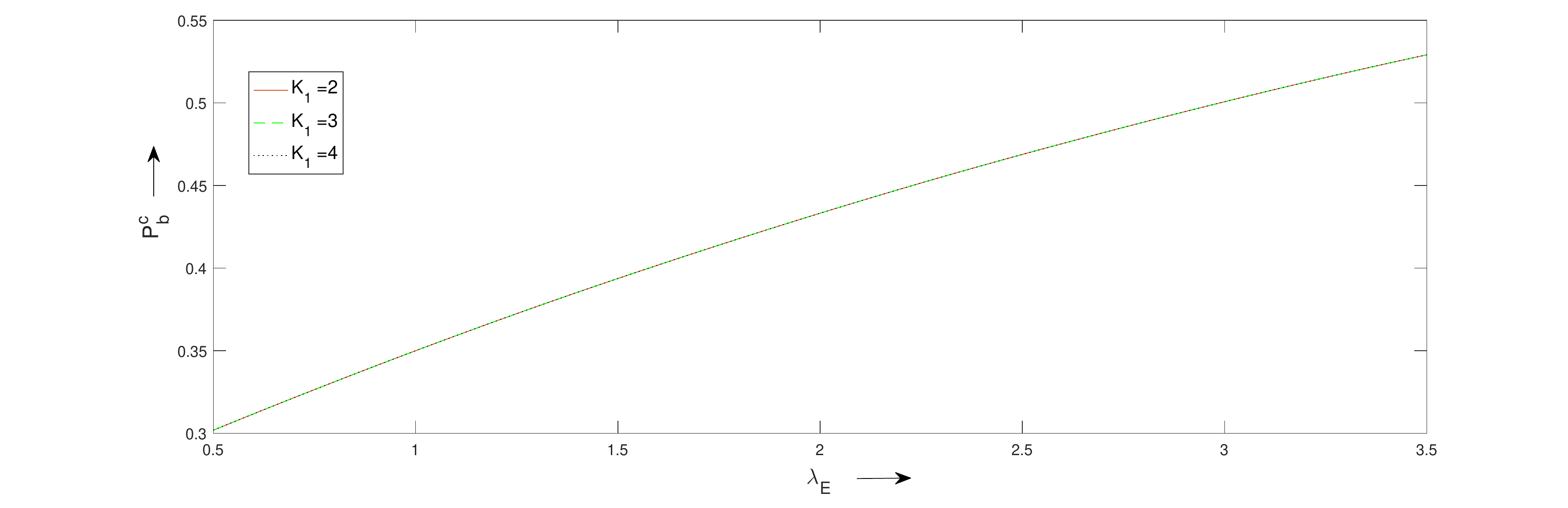}
		\label{fig:6b}}%
		
			\subfigure[$P_{b}^c$ versus  $\mu_{\mathcal{E}}$ for $S=5, K_1=2$]
	{\includegraphics[trim= 3cm 0.1cm 1.1cm 0.4cm, height = 5.55cm,width = 0.5\textwidth]{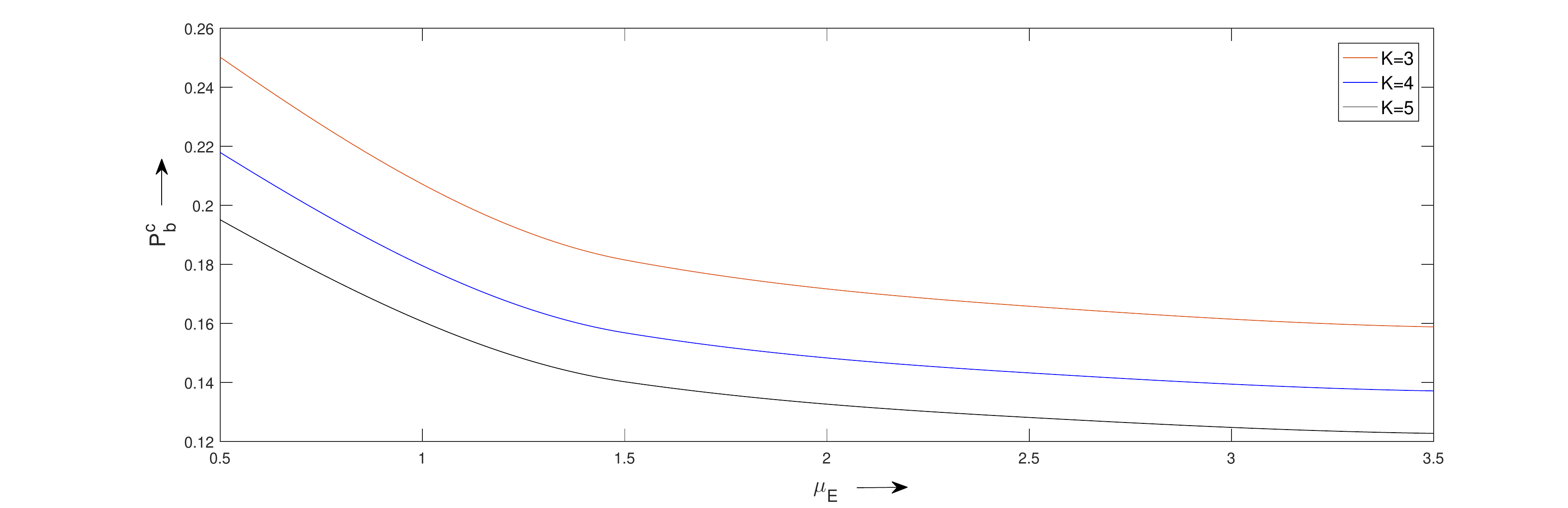}
		\label{fig:6c}}%
	\subfigure[$P_{b}^c$ versus  $\mu_{\mathcal{E}}$ for $S=5,K=4$]
	{\includegraphics[trim= 2cm 0.1cm 2.5cm 0.4cm, height = 5.55cm,width = 0.5\textwidth]{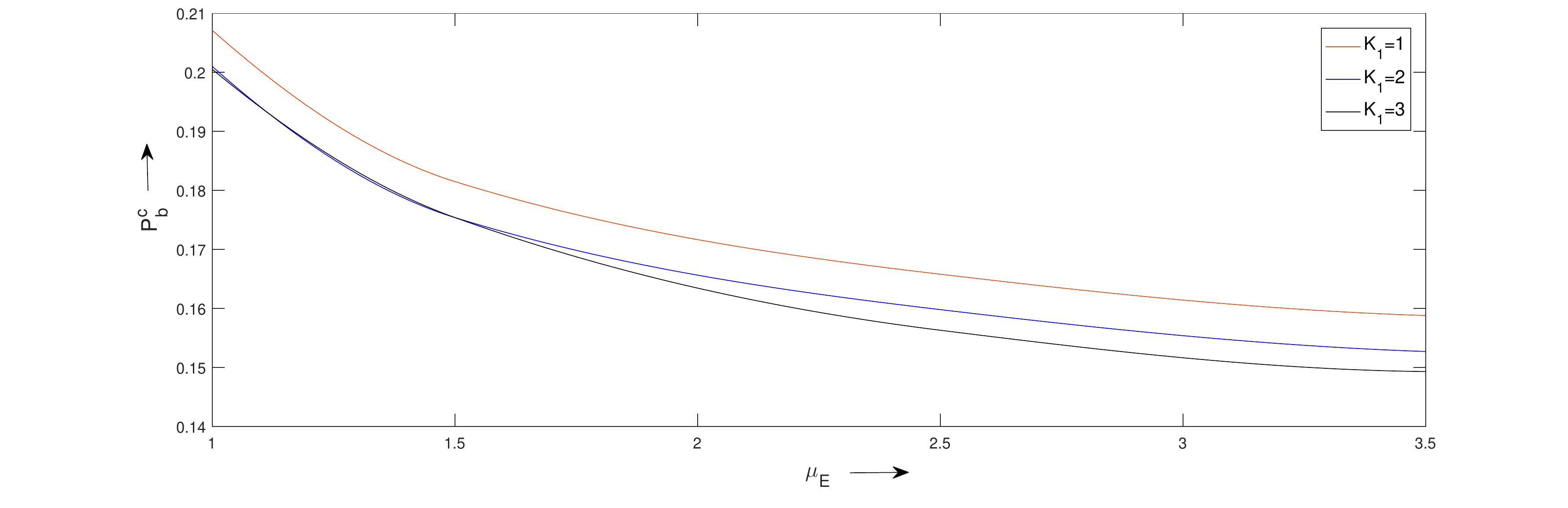}
		\label{fig:6d}}%
			\caption{Dependence of the blocking probability for new call $P_b^c$  over arrival rate of an emergency call $\lambda_{\mathcal{E}}$ and service rate of an emergency call $\mu_{\mathcal{E}}$ for $K_1=2$, and $K=4$,  respectively. }
		
	\label{fig:Pb_Pd6}
\end{figure}
%======================================

%========================================
\begin{figure}[htp]
	\centering
		\subfigure[$P_{preempt}^{emr}$ versus  $\lambda_{\mathcal{E}}$ for $S=5, K_1=2$]
	{\includegraphics[trim= 3cm 0.1cm 1.1cm 0.4cm, height = 5.55cm,width = 0.5\textwidth]{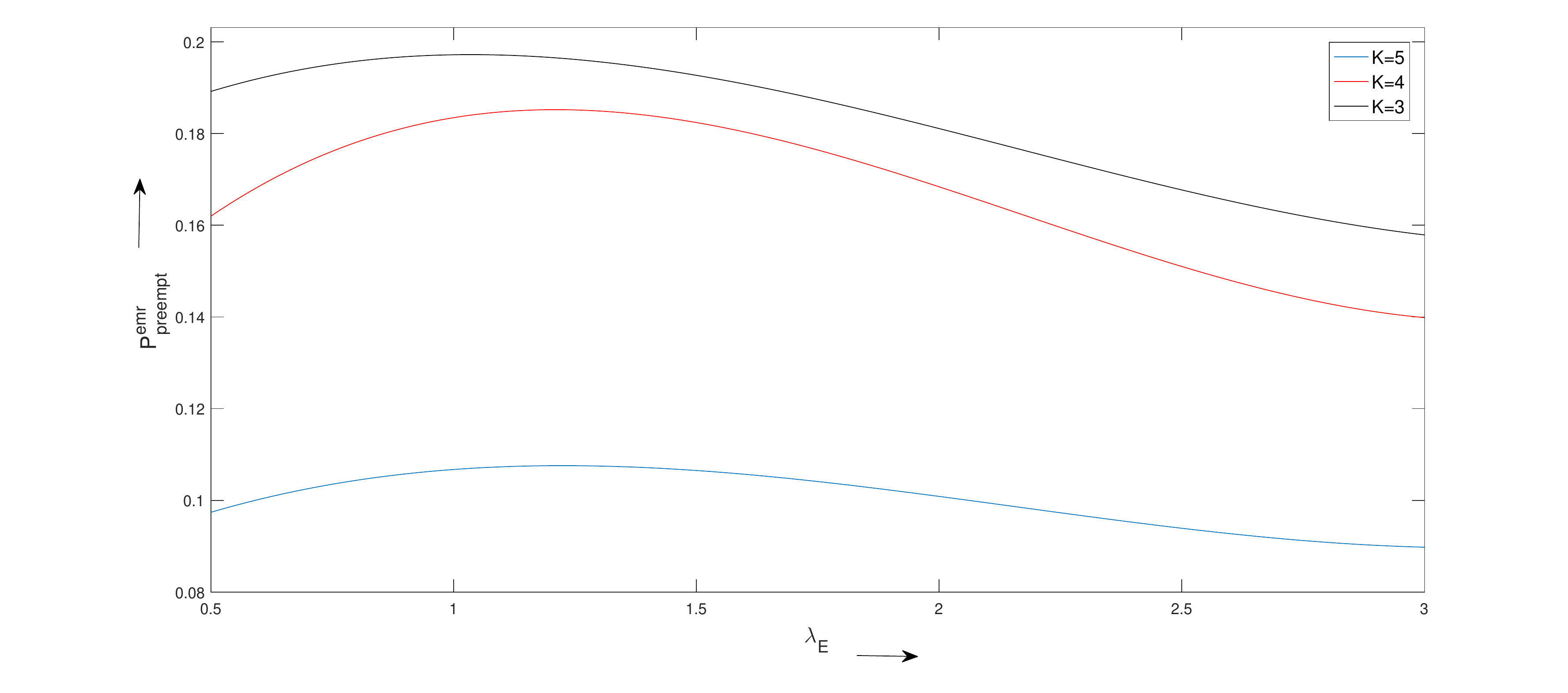}
		\label{fig:7a}}%
	\subfigure[$P_{preempt}^{emr}$ versus  $\lambda_{\mathcal{E}}$ for $S=5,K=4$]
	{\includegraphics[trim= 2cm 0.1cm 2.5cm 0.4cm, height = 5.55cm,width = 0.5\textwidth]{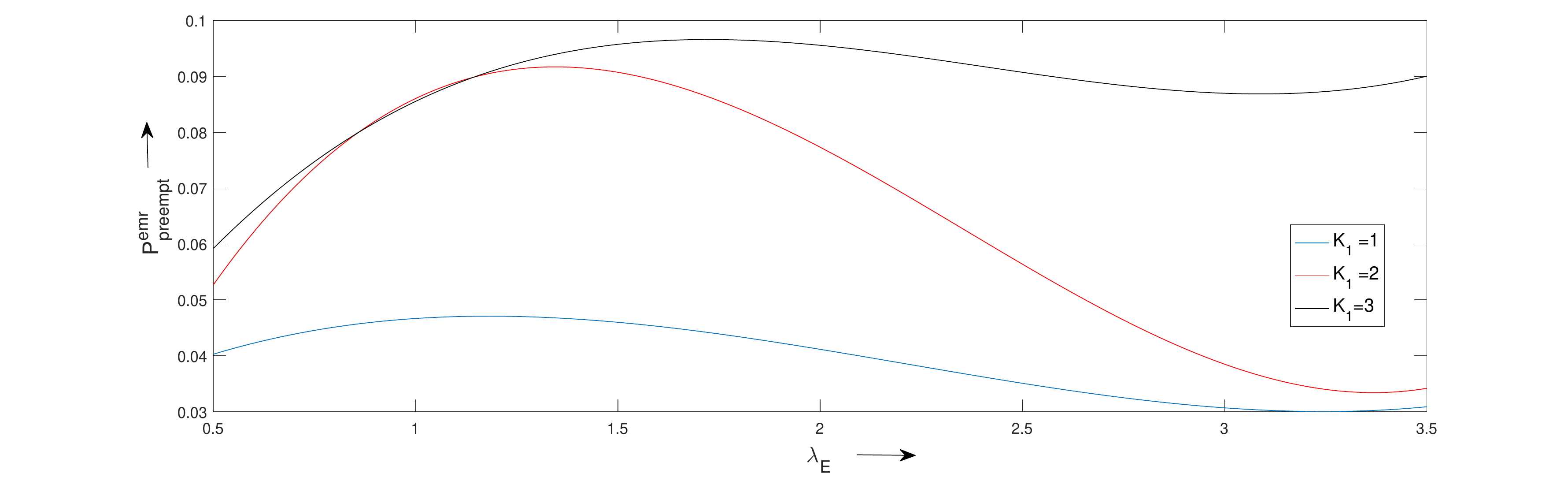}
		\label{fig:7b}}%
		
			\subfigure[$P_{preempt}^{emr}$ versus  $\mu_{\mathcal{E}}$ for $S=5, K_1=2$]
	{\includegraphics[trim= 3cm 0.1cm 1.1cm 0.4cm, height = 5.55cm,width = 0.5\textwidth]{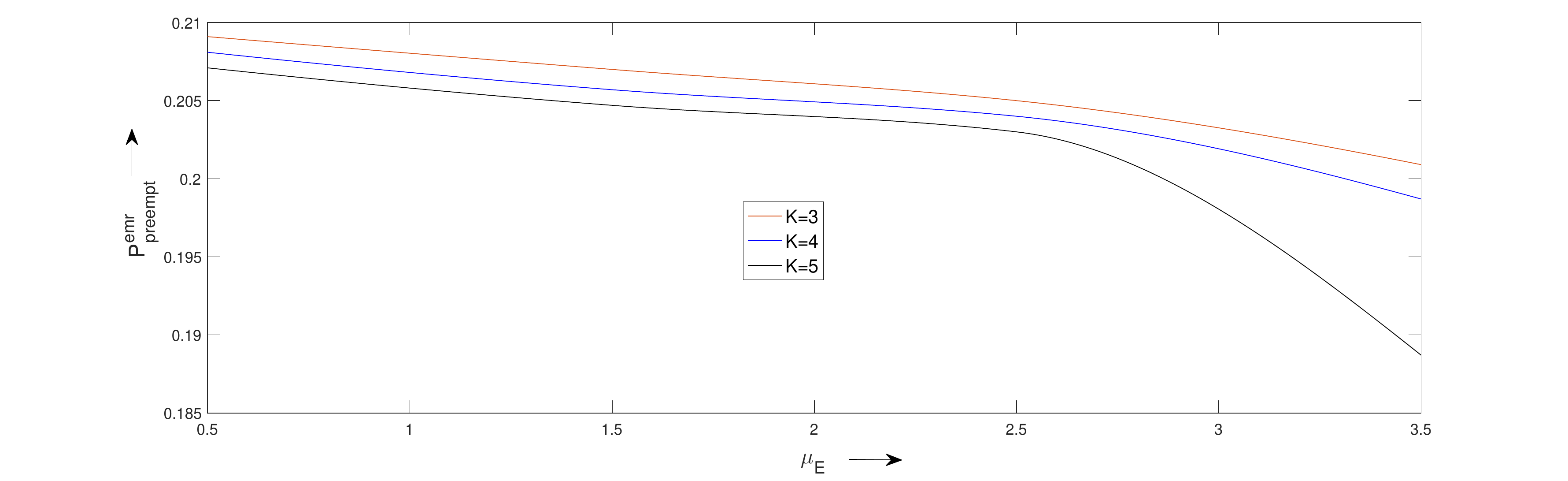}
		\label{fig:7c}}%
	\subfigure[$P_{preempt}^{emr}$ versus  $\mu_{\mathcal{E}}$ for $S=5,K=4$]
	{\includegraphics[trim= 2cm 0.1cm 2.5cm 0.4cm, height = 5.55cm,width = 0.5\textwidth]{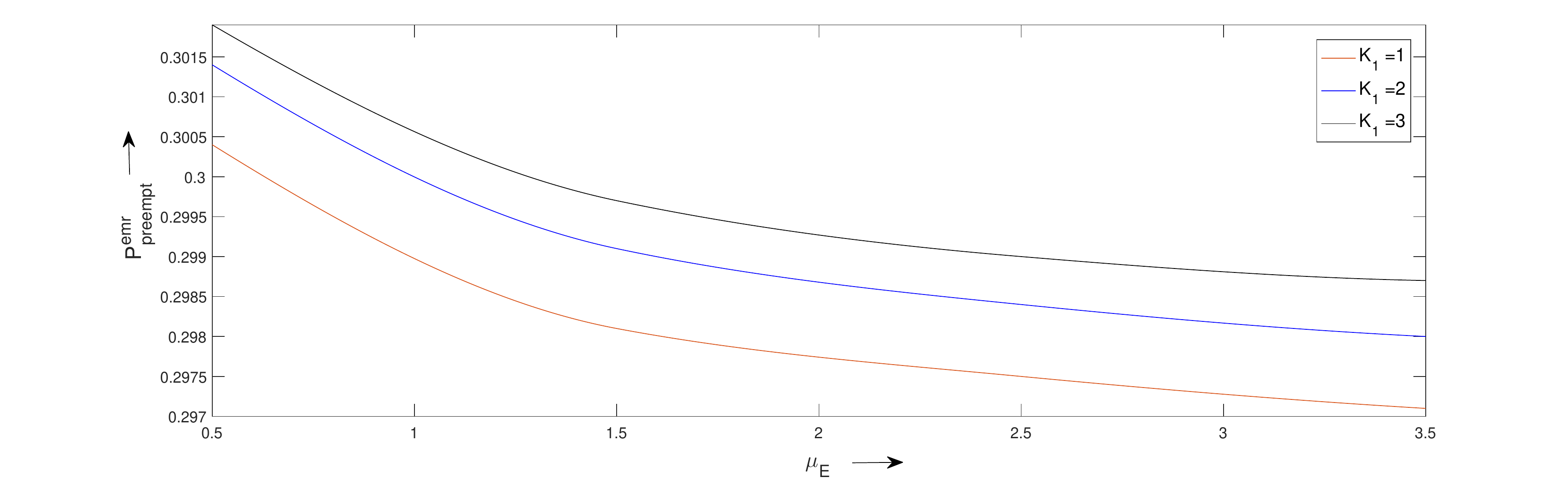}
		\label{fig:7d}}%
			\caption{Dependence of the preemption probability for emergency call $P_{preempt}^{emr}$  over arrival rate $\lambda_{\mathcal{E}}$ and  service rate $\mu_{\mathcal{E}}$ of an emergency call  for $K_1=2$, and $K=4$,  respectively. }
		
	\label{fig:Pb_Pd7}
\end{figure}
%======================================

\section{Optimization Problem} \label{section6}
 In the catastrophic scenario, the loss probabilities, i.e., $P_e$, $P_b^c$ and $P_{preempt}^{emr}$,  should be the performance determining factors for  cellular networks. Any increment in these factors directly  indicates unsatisfactory level of service.  On the other side, it has been  observed from the results (refer, Section \ref{section5}) that these factors are mostly affected by  $\lambda_{\mathcal{E}}$,  $\mu_{\mathcal{E}}$, $K$ and  $K_1$. In order to provide  sufficient backup channels for  service,  an approximated value of $\lambda_{\mathcal{E}}$ will be estimated. Since, in this work, a threshold $K_1$ has been set for the preemption of ongoing new/handoff calls.   Therefore, it is intended to find the optimal values of  $\lambda_{\mathcal{E}}$,  $\mu_{\mathcal{E}}$, $K$ and  $K_1$ such that  loss probabilities should not exceed some pre-defined values targeting the minimum  number of back up channels.  
Such scenario can be modeled by proposing a non-trivial optimization problem with the decision variables $K$, $K_1$, $\mu_{\mathcal{E}}$ and $ \lambda_{\mathcal{E}}$ as follows\\
\begin{center}
$\begin{array}{lll}
&\textrm{min } & \{K, K_1\} \\
&\textrm{subject to},& P_{e}( K,K_1,\lambda_{\mathcal{E}},\mu_{\mathcal{E}}) \leq \epsilon_1,\\
&&P_{b}^{c}(K,K_1,\lambda_{\mathcal{E}},\mu_{\mathcal{E}}) \leq \epsilon_2, \\
&&P_{preempt}^{emr}(K,K_1,\lambda_{\mathcal{E}},\mu_{\mathcal{E}}) \leq \epsilon_3, \\
&& K \leq S,\\
&& K_1 \leq K,\\
&& K,K_1,\lambda_{\mathcal{E}},\mu_{\mathcal{E}} \geq 0.
\end{array}$\\
\end{center}

\noindent Here, $\epsilon_1$, $\epsilon_2$  and $\epsilon_3$ are pre-defined values depending on the tolerance of the system for $P_{e}$, $P_b^c$ and  $P_{preempt}^{emr}$, respectively.  Assume  $\epsilon_1=\epsilon_2=\epsilon_3=10^{-3}$  for the further numerical computation. 
These   constraints are non-linear and highly complex in nature.  Thus, an evolutionary approach, non-dominated sorting genetic algorithm-II (NSGA-II) has been employed to obtain its optimal solution.
The detailed analysis of NSGA-II algorithm can be found in \cite{deb2002fast}. The main steps of NSGA-II are provided as follows 

\begin{table}[htp]
	\centering
	\scalebox{1}{
	\begin{tabular}{|lllllll|}
		\hline
		S=3 & $\mu_{\mathcal{N}}= 1$& $\mu_{\mathcal{H}}= 1$ & $\lambda_{\mathcal{N}}=1$ & $\lambda_{\mathcal{H}}=1$ & $\theta=1$&\\
		\hline
	$K^*$& 	$K_1^*$ &  $\lambda_{\mathcal{E}}^*$& $\mu_{\mathcal{E}}^*$ & $P_{e}$ &$P_b^c$ & $P_{preempt}^{emr}$\\
	\hline
2& 1 & 10.2141 & 12.1111 & 0.00045099 & 0.000975564 & 0.00045766 \\
		\hline
		S=3 & $\mu_{\mathcal{N}}= 1$& $\mu_{\mathcal{H}}= 1.5$ & $\lambda_{\mathcal{N}}=1$ & & $\lambda_{\mathcal{H}}=1$ &$\theta=1$\\
		\hline
	$K^*$& 	$K_1^*$&  $\lambda_{\mathcal{E}}^*$& $\mu_{\mathcal{E}}^*$ & $P_{e}$ &$P_b^c$ & $P_{preempt}^{emr}$\\
		\hline
2& 1 & 11.5447 & 13.0124 & 0.00065478 & 0.0009854 & 0.00045214 \\
		\hline
	S=3 & $\mu_{\mathcal{N}}= 1$& $\mu_{\mathcal{H}}= 2$ & $\lambda_{\mathcal{N}}=1$ & $\lambda_{\mathcal{H}}=1$ &$\theta=1$ & \\
		\hline
	$K^*$& 	$K_1^*$&  $\lambda_{\mathcal{E}}^*$& $\mu_{\mathcal{E}}^*$ & $P_{e}$ &$P_b^c$ & $P_{preempt}^{emr}$\\
		\hline
2& 1 & 12.4338 & 14.1320 & 0.00021099 & 0.00021564 & 0.00013966 \\
		\hline
		S=3 & $\mu_{\mathcal{N}}= 1$& $\mu_{\mathcal{H}}= 2.5$ & $\lambda_{\mathcal{N}}=1$ & $\lambda_{\mathcal{H}}=1$ &$\theta=1$&\\
		\hline
	$K^*$& 	$K_1^*$&  $\lambda_{\mathcal{E}}^*$& $\mu_{\mathcal{E}}^*$ & $P_{e}$ &$P_b^c$ & $P_{preempt}^{emr}$\\
		\hline
2& 1 & 12.0123 & 10.3443 & 0.00041052 & 0.00037981 & 0.00024398 \\
		\hline
		S=3 & $\mu_{\mathcal{N}}= 1$& $\mu_{\mathcal{H}}= 3$ & $\lambda_{\mathcal{N}}=1$ & $\lambda_{\mathcal{H}}=1$ &$\theta=1$&\\
		\hline
	$K^*$& 	$K_1^*$&  $\lambda_{\mathcal{E}}^*$& $\mu_{\mathcal{E}}^*$ & $P_{e}$ &$P_b^c$ & $P_{preempt}^{emr}$\\
		\hline
	2& 2 & 5.2725 & 12.3548 & 0.00087449 & 0.00011405 & 0.00026606 \\
		\hline
		
				S=3 & $\mu_{\mathcal{N}}= 1$& $\mu_{\mathcal{H}}= 3.5$ & $\lambda_{\mathcal{N}}=1$ & $\lambda_{\mathcal{H}}=1$ &$\theta=1$&\\
		\hline
	$K^*$& 	$K_1^*$&  $\lambda_{\mathcal{E}}^*$& $\mu_{\mathcal{E}}^*$ & $P_{e}$ &$P_b^c$ & $P_{preempt}^{emr}$\\
		\hline
	2& 2 & 5.2725 & 12.3548 & 0.00087449 & 0.00011405 & 0.00026606 \\
		\hline
			
				S=3 & $\mu_{\mathcal{N}}= 1$& $\mu_{\mathcal{H}}= 4$ & $\lambda_{\mathcal{N}}=1$ & $\lambda_{\mathcal{H}}=1$ &$\theta=1$&\\
		\hline
	$K^*$& 	$K_1^*$&  $\lambda_{\mathcal{E}}^*$& $\mu_{\mathcal{E}}^*$ & $P_{e}$ &$P_b^c$ & $P_{preempt}^{emr}$\\
		\hline
	2& 2 & 5.2725 & 12.3548 & 0.00087449 & 0.00011405 & 0.00026606 \\
		\hline
				\hline
	S=4 & $\mu_{\mathcal{N}}= 1$& $\mu_{\mathcal{H}}= 1$ & $\lambda_{\mathcal{N}}=1$ & $\lambda_{\mathcal{H}}=1$ &$\theta=1$&\\
		\hline
	$K^*$& 	$K_1^*$&  $\lambda_{\mathcal{E}}^*$& $\mu_{\mathcal{E}}^*$ & $P_{e}$ &$P_b^c$ & $P_{preempt}^{emr}$\\
	\hline
3& 2 & 6.0571 & 8.1215 & 0.00048109 & 0.00093954 & 0.00084930 \\
		\hline
				S=4 & $\mu_{\mathcal{N}}= 1$& $\mu_{\mathcal{H}}= 1.5$ & $\lambda_{\mathcal{N}}=1$ & $\lambda_{\mathcal{H}}=1$ &$\theta=1$&\\
		\hline
	$K^*$& 	$K_1^*$&  $\lambda_{\mathcal{E}}^*$& $\mu_{\mathcal{E}}^*$ & $P_{e}$ &$P_b^c$ & $P_{preempt}^{emr}$\\
	\hline
2& 1 & 7.6398 & 6.5765 & 0.00022035 & 0.00019028 & 0.00015352 \\
		\hline
			S=4 & $\mu_{\mathcal{N}}= 1$& $\mu_{\mathcal{H}}= 2$ & $\lambda_{\mathcal{N}}=1$ & $\lambda_{\mathcal{H}}=1$ &$\theta=1$&\\
		\hline
	$K^*$& 	$K_1^*$&  $\lambda_{\mathcal{E}}^*$& $\mu_{\mathcal{E}}^*$ & $P_{e}$ &$P_b^c$ & $P_{preempt}^{emr}$\\
	\hline
3& 2 & 8.3102 & 9.8119 & 0.00032587 & 0.00083375 & 0.00011802 \\
		\hline
			S=4 & $\mu_{\mathcal{N}}= 1$& $\mu_{\mathcal{H}}= 2.5$ & $\lambda_{\mathcal{N}}=1$ & $\lambda_{\mathcal{H}}=1$ &$\theta=1$&\\
		\hline
	$K^*$& 	$K_1^*$&  $\lambda_{\mathcal{E}}^*$& $\mu_{\mathcal{E}}^*$ & $P_{e}$ &$P_b^c$ & $P_{preempt}^{emr}$\\
	\hline
3& 1 & 9.5971 & 9.8364 & 0.00089972 & 0.00013318 & 0.00010531 \\
		\hline
			S=4 & $\mu_{\mathcal{N}}= 1$& $\mu_{\mathcal{H}}= 3$ & $\lambda_{\mathcal{N}}=1$ & $\lambda_{\mathcal{H}}=1$ &$\theta=1$&\\
		\hline
	$K^*$& 	$K_1^*$&  $\lambda_{\mathcal{E}}^*$& $\mu_{\mathcal{E}}^*$ & $P_{e}$ &$P_b^c$ & $P_{preempt}^{emr}$\\
	\hline
3& 1 & 7.3306 & 8.4685 & 0.00019782 & 0.00013355 & 0.00016431 \\
		\hline
		S=4 & $\mu_{\mathcal{N}}= 1$& $\mu_{\mathcal{H}}= 3.5$ & $\lambda_{\mathcal{N}}=1$ & $\lambda_{\mathcal{H}}=1$ &$\theta=1$&\\
		\hline
	$K^*$& 	$K_1^*$&  $\lambda_{\mathcal{E}}^*$& $\mu_{\mathcal{E}}^*$ & $P_{e}$ &$P_b^c$ & $P_{preempt}^{emr}$\\
		\hline
		2& 2 & 7.6404 & 7.0145 & 0.00048684 & 0.00010886 & 0.00060177 \\
			\hline
			S=4 & $\mu_{\mathcal{N}}= 1$& $\mu_{\mathcal{H}}= 4$ & $\lambda_{\mathcal{N}}=1$ & $\lambda_{\mathcal{H}}=1$ &$\theta=1$&\\
		\hline
	$K^*$& 	$K_1^*$&  $\lambda_{\mathcal{E}}^*$& $\mu_{\mathcal{E}}^*$ & $P_{e}$ &$P_b^c$ & $P_{preempt}^{emr}$\\
		\hline
		3& 2 & 6.4143 & 9.7242 & 0.00029148 & 0.00062101 & 0.00083338 \\

		\hline
	\end{tabular}}
	\caption{Optimal values of $\mu_{\mathcal{E}}^*$, $\lambda_{\mathcal{E}}^*$, $K^*$ and $K_1^*$ for different values of $S$ by applying NSGA-II method.}
	\label{tab:my_label1}
\end{table}

\begin{table}[htp]
	\centering
	\scalebox{1}{
	\begin{tabular}{|lllllll|}
		\hline
		S=5 & $\mu_{\mathcal{N}}= 1$& $\mu_{\mathcal{H}}= 1$ & $\lambda_{\mathcal{N}}=1$ & $\lambda_{\mathcal{H}}=1$ & $\theta=1$&\\
		\hline
	$K^*$& 	$K_1^*$ &  $\lambda_{\mathcal{E}}^*$& $\mu_{\mathcal{E}}^*$ & $P_{e}$ &$P_b^c$ & $P_{preempt}^{emr}$\\
	\hline
	4 & 1 & 5.3121 & 4.1682 & 0.0004323 & 0.00010041 & 0.00010210 \\
		\hline
		S=5 & $\mu_{\mathcal{N}}= 1$& $\mu_{\mathcal{H}}= 1.5$ & $\lambda_{\mathcal{N}}=1$ & & $\lambda_{\mathcal{H}}=1$ &$\theta=1$\\
		\hline
	$K^*$& 	$K_1^*$&  $\lambda_{\mathcal{E}}^*$& $\mu_{\mathcal{E}}^*$ & $P_{e}$ &$P_b^c$ & $P_{preempt}^{emr}$\\
		\hline
	4 & 2 & 6.6293 & 8.238 & 0.0008641 & 0.00028218 & 0.0003535 \\

		\hline
	S=5 & $\mu_{\mathcal{N}}= 1$& $\mu_{\mathcal{H}}= 2$ & $\lambda_{\mathcal{N}}=1$ & $\lambda_{\mathcal{H}}=1$ &$\theta=1$ & \\
		\hline
	$K^*$& 	$K_1^*$&  $\lambda_{\mathcal{E}}^*$& $\mu_{\mathcal{E}}^*$ & $P_{e}$ &$P_b^c$ & $P_{preempt}^{emr}$\\
		\hline
	4 & 2 & 7.4782 & 7.2992 & 0.0001746 & 0.00043946 & 0.00013787 \\

		\hline
		S=5 & $\mu_{\mathcal{N}}= 1$& $\mu_{\mathcal{H}}= 2.5$ & $\lambda_{\mathcal{N}}=1$ & $\lambda_{\mathcal{H}}=1$ &$\theta=1$&\\
		\hline
	$K^*$& 	$K_1^*$&  $\lambda_{\mathcal{E}}^*$& $\mu_{\mathcal{E}}^*$ & $P_{e}$ &$P_b^c$ & $P_{preempt}^{emr}$\\
		\hline
	4 & 2 & 6.9123 & 8.8017 & 0.0005551 & 0.00020728 & 0.0003534 \\

		\hline
		S=5 & $\mu_{\mathcal{N}}= 1$& $\mu_{\mathcal{H}}= 3$ & $\lambda_{\mathcal{N}}=1$ & $\lambda_{\mathcal{H}}=1$ &$\theta=1$&\\
		\hline
	$K^*$& 	$K_1^*$&  $\lambda_{\mathcal{E}}^*$& $\mu_{\mathcal{E}}^*$ & $P_{e}$ &$P_b^c$ & $P_{preempt}^{emr}$\\
		\hline

	4 & 2 & 7.6069 & 6.9782 & 0.0005735 & 0.00033208 & 0.0005638 \\

			\hline
	S=5 & $\mu_{\mathcal{N}}= 1$& $\mu_{\mathcal{H}}= 3.5$ & $\lambda_{\mathcal{N}}=1$ & $\lambda_{\mathcal{H}}=1$ &$\theta=1$&\\
		\hline
	$K^*$& 	$K_1^*$&  $\lambda_{\mathcal{E}}^*$& $\mu_{\mathcal{E}}^*$ & $P_{e}$ &$P_b^c$ & $P_{preempt}^{emr}$\\
	\hline
	3 & 3 & 5.543 & 7.4075 & 0.00066154 & 0.00080312 & 0.00045179 \\

		\hline
				S=5 & $\mu_{\mathcal{N}}= 1$& $\mu_{\mathcal{H}}= 4$ & $\lambda_{\mathcal{N}}=1$ & $\lambda_{\mathcal{H}}=1$ &$\theta=1$&\\
		\hline
	$K^*$& 	$K_1^*$&  $\lambda_{\mathcal{E}}^*$& $\mu_{\mathcal{E}}^*$ & $P_{e}$ &$P_b^c$ & $P_{preempt}^{emr}$\\
	\hline
		4 & 1 & 5.5273 & 7.034 & 0.00013459 & 0.00012593 & 0.00096669 \\
	\hline
		\hline
		S=6 & $\mu_{\mathcal{N}}= 1$& $\mu_{\mathcal{H}}= 1$ & $\lambda_{\mathcal{N}}=1$ & $\lambda_{\mathcal{H}}=1$ &$\theta=1$&\\
		\hline
	$K^*$& 	$K_1^*$&  $\lambda_{\mathcal{E}}^*$& $\mu_{\mathcal{E}}^*$ & $P_{e}$ &$P_b^c$ & $P_{preempt}^{emr}$\\
	\hline
	5& 2 & 6.9123 & 8.8017 & 0.00023095 & 0.00023364 & 0.0005995 \\
	\hline

			S=6 & $\mu_{\mathcal{N}}= 1$& $\mu_{\mathcal{H}}= 1.5$ & $\lambda_{\mathcal{N}}=1$ & $\lambda_{\mathcal{H}}=1$ &$\theta=1$&\\
		\hline
	$K^*$& 	$K_1^*$&  $\lambda_{\mathcal{E}}^*$& $\mu_{\mathcal{E}}^*$ & $P_{e}$ &$P_b^c$ & $P_{preempt}^{emr}$\\
	
	\hline
4& 1 & 7.4402 & 7.2137 & 0.00076359 & 0.00090869 & 0.0004349 \\
		\hline
			S=6 & $\mu_{\mathcal{N}}= 1$& $\mu_{\mathcal{H}}= 2$ & $\lambda_{\mathcal{N}}=1$ & $\lambda_{\mathcal{H}}=1$ &$\theta=1$&\\
		\hline
	$K^*$& 	$K_1^*$&  $\lambda_{\mathcal{E}}^*$& $\mu_{\mathcal{E}}^*$ & $P_{e}$ &$P_b^c$ & $P_{preempt}^{emr}$\\
	\hline
4& 4 & 5.4915 & 6.2046 & 0.00042210 & 0.00048330 & 0.00035433 \\
		\hline
			S=6 & $\mu_{\mathcal{N}}= 1$& $\mu_{\mathcal{H}}= 2.5$ & $\lambda_{\mathcal{N}}=1$ & $\lambda_{\mathcal{H}}=1$ &$\theta=1$&\\
		\hline
	$K^*$& 	$K_1^*$&  $\lambda_{\mathcal{E}}^*$& $\mu_{\mathcal{E}}^*$ & $P_{e}$ &$P_b^c$ & $P_{preempt}^{emr}$\\
	\hline
4& 2 & 8.5686 & 8.0211 & 0.00013600 & 0.00067048 & 0.00055984 \\
		\hline
			S=6 & $\mu_{\mathcal{N}}= 1$& $\mu_{\mathcal{H}}= 3$ & $\lambda_{\mathcal{N}}=1$ & $\lambda_{\mathcal{H}}=1$ &$\theta=1$&\\
		\hline
	$K^*$& 	$K_1^*$&  $\lambda_{\mathcal{E}}^*$& $\mu_{\mathcal{E}}^*$ & $P_{e}$ &$P_b^c$ & $P_{preempt}^{emr}$\\
	\hline
4& 2 & 8.4417 & 8.8241 & 0.00024150 & 0.00045178 & 0.00047851 \\
		\hline
	
			S=6 & $\mu_{\mathcal{N}}= 1$& $\mu_{\mathcal{H}}= 3.5$ & $\lambda_{\mathcal{N}}=1$ & $\lambda_{\mathcal{H}}=1$ &$\theta=1$&\\
		\hline
	$K^*$& 	$K_1^*$&  $\lambda_{\mathcal{E}}^*$& $\mu_{\mathcal{E}}^*$ & $P_{e}$ &$P_b^c$ & $P_{preempt}^{emr}$\\
	\hline
	4& 2 & 8.41000 & 8.7895 & 0.00046985 & 0.00098754 & 0.00044165 \\
		\hline
			S=6 & $\mu_{\mathcal{N}}= 1$& $\mu_{\mathcal{H}}= 4$ & $\lambda_{\mathcal{N}}=1$ & $\lambda_{\mathcal{H}}=1$ &$\theta=1$&\\
		\hline
	$K^*$& 	$K_1^*$&  $\lambda_{\mathcal{E}}^*$& $\mu_{\mathcal{E}}^*$ & $P_{e}$ &$P_b^c$ & $P_{preempt}^{emr}$\\
	\hline
	4& 2 & 8.7854 & 8.7455 & 0.00087458 & 0.00044698 & 0.00065458 \\
		\hline
	\end{tabular}}
	\caption{Optimal values of $\mu_{\mathcal{E}}^*$, $\lambda_{\mathcal{E}}^*$, $K^*$ and $K_1^*$ for different values of $S$ by applying NSGA-II method.}
	\label{tab:my_label2}
\end{table}

\begin{itemize}
    \item[1.] \textbf{Initialization:} Initialize the population size $P$ based on the number of decision variables.
    \item[2.] \textbf{Non-dominated Sorting:} The initialized population is sorted on the basis of non-domination. Each solution is assigned a fitness or a rank equal to its non-domination level. Steps of sort algorithm are as follows
    \begin{itemize}
        \item Initialize $S_p= \phi$, where $S_p$ is set of all individuals dominated by $p$, where $p \in P$.
        \item Initialize $n_p=0$.  This is the number of individuals that dominate $p$.
        \item For each individual $q\in P$, if $p$ dominates $q$, add $q$ to the set $S_p$, i.e., $S_p=S_p \cup \{q\} $ else  increment the domination counter for $p$, i.e., $n_p=n_p+1$.
        \item Initialize $F_1 = \phi$ where $F_1$ is the first front.
        \item If $n_p=0$,  $p \in F_1$. Set rank of individual $p$ to 1. Update  $F_1 = F_1 \cup \{p\}.$
        This will be  carried out for all the individuals in the  main population $P.$ 
        \item Initialize $i=1$, where $i$ denotes the front counter. 
        \item Define $Q=\phi$, where $Q$ is the set for storing the individuals  for $(i+1)^{th}$ front. For each individual $q$ in $S_p$, if $q$ dominates, set $n_q=n_q-1$, decrements the domination count for individual $q$. If  $n_q=0$, then none of the individuals in the subsequent fronts would dominate $q$. Hence, set $q_{rank}=i+1$.  Update the set $Q=Q \cup q$.
        \item Increment the front counter by one and set the next front as $F_i=Q$. This step is carried out while the $i^{th}$ front is non empty.
    \end{itemize}
    \item[3.] \textbf{Crowding Distance:} All the individuals after non-dominated sort are assigned a crowding distance value. Crowding distance is assigned front wise and compared between two individuals.
    \item[4.] \textbf{Selection:} Once the individuals are sorted based on non-domination and with crowding distance assigned, the selection is carried out using a crowded-comparison operator.
    \item[5.] \textbf{Recombination and Selection:} The offspring population is combined with the current generation population and selection is performed to set the individuals of the next generation. 
\end{itemize}

 Tables \ref{tab:my_label1} and \ref{tab:my_label2} represent the optimal values of $K^*$, $K_1^*$, $\lambda_{\mathcal{E}}^*$, and $\mu_{\mathcal{E}}^*$ for different combinations of  service rate of handoff call and different values of $S$.  All results are obtained by MATLAB software, which are run on a computer with Intel Core i7-6700 3.40GHz CPU and 8 GB of RAM. The obtained results provide the value of optimal backup channels for various combinations of  total number of channels and arrival rates. The proposed optimization problem's sensitivity analysis is useful in estimating the number of backup channels in an emergency scenario. It can be observed from the obtained results that when  $S=3$; $K=66\%$ of $S$, $S=4$; $K=77\%$ of $S$, $S=5$; $K=80\%$ of $S$, and $S=6$; $K=66\%$ of $S$, backup channels are required taking into account that all the loss probabilities in the catastrophic scenario remains under  pre-defined small values. Though, it is impossible to provide the same level of service and maintain the same number of resources in an emergency situation as in a normal situation, some fixed resources can be preserved as a backup. These results have the potential to be tremendously useful in communication systems and cellular networks.

\section{Conclusions} \label{section7}
In cutting-edge wireless technology, queuing models with catastrophic events are a driving force in communications and cellular networks. Varying classes of traffic, including as video, audio, pictures, data, and so on, are ascribed different levels of importance in these sorts of catastrophic queueing models,  and consequently their services are effectuated in accordance with an appropriate priority policy.  In cellular networks, the kinds of systems where a higher priority traffic has an advantage in access to service compared to less important ones, are explored through priority policies. Therefore, this study  explores a $\textrm{\it MMAP[c]/PH[c]/S}$  catastrophic  queueing model with controllable preemptive repeat priority policy and \textrm{$P\!H$} distributed retrial times. Due to the brief span of inter-retrial times in comparison to service times, a more generalized approach, \textrm{$P\!H$} distributed retrial
times is used so that the performance of the system is not over or under estimated. 
When a calamity strikes, backup channels are used to establish communication in the affected area. The underlying process of the presented system is modeled by \textrm{$A\!Q\!T\!M\!C$}. Ergodicity conditions of the underlying Markov chain are obtained by proving that the Markov chain belongs to the class of \textrm{$A\!Q\!T\!M\!C$}. 
A modified algorithm is applied  for approximate computation of the stationary distribution. 
To establish the communication in case of calamity,  estimation of the number of backup channels with respect to the total number of channels is very important. Therefore a multi-objective optimization problem to obtain optimal value of total number of  backup channels and threshold level for preemption  have been formulated and dealt by employing NSGA-II approach. These findings could be  beneficial in communication systems and cellular networks. In the future, authors propose to
 extend this model by using the  preemptive resume priority policy for such multi-server queueing model.

\end{document}